\documentclass[11pt]{amsart}
\usepackage{amssymb,mathrsfs,graphicx,enumerate}
\usepackage{amsmath,amsfonts,amssymb,amscd,amsthm,bbm}
\usepackage[retainorgcmds]{IEEEtrantools}
\usepackage{colortbl}

\title[A gradient flow approach for the Lohe matrix model]{A gradient flow formulation of the Lohe matrix model with a high-order polynomial coupling}

\author[Ha]{Seung-Yeal Ha}
\address[Seung-Yeal Ha]{\newline Department of Mathematical Sciences\newline Seoul National University, Seoul 08826, and \newline
Korea Institue for Advanced Study, Hoegiro 85, 02455, Seoul, Republic of Korea}
\email{syha@snu.ac.kr}

\author[Park]{Hansol Park}
\address[Hansol Park]{\newline Department of Mathematical Sciences\newline Seoul National University, Seoul 08826, Republic of Korea}
\email{hansol960612@snu.ac.kr}

\newtheorem{theorem}{Theorem}[section]
\newtheorem{lemma}{Lemma}[section]
\newtheorem{corollary}{Corollary}[section]
\newtheorem{proposition}{Proposition}[section]
\newtheorem{remark}{Remark}[section]

\newtheorem{definition}{Definition}[section]

\newcommand{\bbr}{\mathbb R}

\newcommand{\bbu}{\mathbb U}

\newcommand{\bbc}{\mathbb C}

\begin{document}

\date{\today}


\thanks{\textbf{Acknowledgment.} The work of S.-Y.Ha is supported by NRF-2020R1A2C3A01003881}

\begin{abstract}
We present a generalized Lohe matrix model for a homogeneous ensemble with higher order couplings via the gradient flow approach. For the homogeneous free flow with the same hamiltonian, it is well known that the Lohe matrix model with cubic couplings can recast as a gradient system with a potential which is a squared Frobenius norm of  of averaged state. In this paper, we further derive a generalized Lohe matrix model with higher-order couplings via gradient flow approach for a polynomial potential. For the proposed model, we also provide a sufficient framework in terms of coupling strengths and initial data, which  leads to the emergent dynamics of the homogeneous ensemble.
\end{abstract}

\maketitle \centerline{\date}

%\tableofcontents

\section{Introduction} \label{sec:1}
\setcounter{equation}{0}
Synchronous dynamics of oscillatory systems often appears naturally in our daily life, e.g., synchronous heart beating \cite{Pe} and synchronous firing of fireflies \cite{A-B,B-B, P-R, St, Wi1}, etc. Then, one of natural questions would be a mathematical model which exhibits collective synchronous behaviors. In this direction, Arthur Winfree \cite{Wi2} and Yoshiki Kuramoto \cite{Ku2} proposed analytically manageable simple mathematical models in a half century ago, and they provided a sufficient framework leading to the emergent dynamics of weakly coupled oscillators. Recently, the authors introduced a generalized aggregation model on the space of tensors, namely ``{\it the Lohe tensor model}" which encompasses all the previous Lohe type aggregation models such as the Kuramoto model \cite{B-C-M, C-H-J-K, C-S, D-X,  D-B, D-B1, H-K-R1, H-L-X, Ku1, Ku2, V-M1, V-M2}, the Lohe sphere model \cite{C-H5, O} and the Lohe matrix model \cite{D, H-K-R0,  H-R, Lo-6, Lo-5, Lo-1, Lo-2}. However, all interaction mechanism in  aforementioned models are given by cubic couplings (see Section \ref{sec:2} for details). 

In this paper, we are mainly interested in the generalization of the Lohe matrix model with high-order couplings. More precisely, let $U_j = U_j(t)$ be a $d\times d$ complex matrix, and the dynamics of the state matrix is given by the first-order matrix-valued continuous dynamical system.
\begin{equation} \label{A-1}
{\mathrm i}\dot{U}_jU_j^\dagger=H_j+ \frac{{\mathrm i}\kappa}{2N}\sum_{k=1}^N\left( U_kU_j^\dagger -U_jU_k^\dagger \right), \quad t >0, \quad  j = 1, \cdots, N,
\end{equation}
where $U^{\dagger}_j$ is the Hermitian conjugation of $U_j$. Then, it is easy to see that the quadratic quantity $U_j U_j^{\dagger}$ is conserved along \eqref{A-1}. Hence, system \eqref{A-1} can be rewritten as follows:
\begin{equation} \label{A-2}
\dot{U}_j =-{\mathrm i} H_j U_j + \frac{\kappa}{2N}\sum_{k=1}^N\left( U_kU_j^\dagger U_j -U_jU_k^\dagger U_j \right), \quad j = 1, \cdots, N.
\end{equation}
or equivalently
\begin{equation} \label{A-3}
\dot{U}_j =-{\mathrm i} H_j U_j + \frac{\kappa}{2N}\sum_{k=1}^N\left( U_k -U_jU_k^\dagger U_j \right), \quad j = 1, \cdots, N.
\end{equation}
The above system was first introduced in \cite{Lo-1, Lo-2} without resorting on the first principle as one of possible non-abelian generalizations of the Kuramoto model. In what follows, we are interested in the following two questions: \newline
\begin{itemize}
\item
(Q1):~First of all, why only the cubic couplings is involved in the R.H.S. of the Lohe matrix model \eqref{A-2}?
\vspace{0.05cm}

\item
(Q2):~If cubic couplings are not essential,  what kind of couplings can be allowed for the formation of aggregation?
\end{itemize}

\vspace{0.2cm}

The model was not introduced by a hamiltonian formalism or variational approach. Hence it is not clear why  cubic couplings were involved in. As briefly discussed in \cite{H-P3}, any odd order of couplings will be possible for the Lohe tensor model. In a recent paper by the first author and his collaborators, the Lohe matrix model with the same Hamiltonian $H_j  = H$ can recast as a gradient flow. Thanks to solution splitting property for the Lohe matrix model, we can set $H = 0$ without loss of generality. \newline

Consider the equivalent form of the Lohe matrix model:
\begin{equation} \label{A-4}
\dot{U}_j = \frac{\kappa}{2N}\sum_{k=1}^N\left( U_kU_j^\dagger U_j -U_jU_k^\dagger U_j \right), \quad j = 1, \cdots, N.
\end{equation}
A gradient flow formulation of \eqref{A-4} was first investigated in \cite{H-K-R0}. For a homogeneous Lohe matrix ensemble $\{U_j \}$, we set 
\begin{equation*} \label{A-5}
U_c : = \frac{1}{N} \sum_{j=1}^{N} U_j, \quad {\mathcal V}_1(U) := -\frac{\kappa N}{2} \| U_c \|_F^2. 
\end{equation*}
Then, system \eqref{A-3} can be recast as a gradient flow (Proposition \ref{P2.3}):
\begin{equation*} \label{A-6}
{\dot U}_i = -\left.\frac{\partial \mathcal{V}_1}{\partial U_i}\right|_{T_{U_i}\bbu(d)}, \quad i = 1, \cdots, N.
\end{equation*}
Main question is how to define a potential suitably so that the resulting gradient flow exhibits an emergent aggregation dynamics. \newline

Next, we briefly discuss our main results. First, we introduce a potential as follows. For $m \geq 1$, we set 
\begin{equation*} \label{A-7}
 {\mathcal V}_m(U) :=-\frac{\kappa N}{2m}\mathrm{tr}((U_cU_c^\dagger)^m)=-\frac{\kappa N}{2m}\mathrm{tr}(\underbrace{U_cU_c^\dagger \cdots U_cU_c^\dagger}_{2m}).
\end{equation*} 
Note that it follows from the property of trace: $\mathrm{tr}[AB]=\mathrm{tr}[BA]$ that the potential ${\mathcal V}_m(U)$ can be rewritten as 
\[
\mathrm{tr}[(U_cU_c^\dagger)^m] =\mathrm{tr}[(U_c^\dagger U_c)^m].
\]
So actually $U_c^\dagger U_c$ and $U_c U_c^\dagger$ have same effect in the trace function. \newline

Then the corresponding gradient flow
\[  \dot{U}_j = -\left. \frac{\partial\mathcal{V}_m}{\partial U_j}\right|_{T_{U_j} \mathbb{U}(d)}, \quad j = 1, \cdots, N. \]
can be expressed as 
\begin{align}
\begin{aligned} \label{A-8}
& {\mathrm i} \dot{U}_j  U_j^{\dagger} = H_j  +  \frac{{\mathrm i} \kappa}{2N^{2m-1}} \\
& \hspace{0.5cm} \times \sum_{k_1, \cdots, k_{2m-1} = 1}^{N} \Big(  U_{k_1} U_{k_2}^{\dagger} \cdots U_{k_{2m-2}}^{\dagger} U_{k_{2m-1}} U_j^{\dagger}  - U_j  U_{k_{2m-1}}^{\dagger} U_{k_{2m-2}} \cdots U_{k_2} U_{k_1}^{\dagger}  \Big).
\end{aligned}
\end{align}
Then, it is easy to see that $U_j U_j^{\dagger}$ is a conserved quantity for \eqref{A-8}. Thus for $U_j U_j^{\dagger}  = I_d$, system \eqref{A-8} can be further rewritten as 
\begin{align}
\begin{aligned} \label{A-9}
& \dot{U}_j  = -{\mathrm i}  H_j U_j  +  \frac{\kappa}{2N^{2m-1}} \\
& \hspace{0.5cm} \times  \sum_{k_1, \cdots, k_{2m-1} = 1}^{N} \Big(  U_{k_1} U_{k_2}^{\dagger} \cdots U_{k_{2m-2}}^{\dagger} U_{k_{2m-1}} U_j^{\dagger} U_j  - U_j  U_{k_{2m-1}}^{\dagger} U_{k_{2m-2}} \cdots U_{k_2} U_{k_1}^{\dagger}  U_j \Big).
\end{aligned}
\end{align}
Note that the order in the coupling term in the R.H.S. of \eqref{A-9} is $2m +1$. Second, we provide a general case:
\begin{equation} \label{A-10}
{\mathrm i} \dot{U}_j U_j^{\dagger} = H_j +  \sum_{n=1}^m\frac{{\mathrm i} \kappa_n}{2}(\underbrace{U_cU_c^\dagger U_c\cdots U_c^\dagger U_c}_{2n-1} U_j^{\dagger}-U_j \underbrace{U_c^\dagger U_cU_c^\dagger \cdots U_cU_c^\dagger}_{2n-1}).
\end{equation}
It is easy to see that the R.H.S. of \eqref{A-10} is skew-hermitian so that system \eqref{A-10} conserves the quantity $U_j U_j^{\dagger}$. For an ensemble, we set 
\begin{equation*} \label{A-11}
\mathcal{V}_{poly}:=-N\mathrm{tr}(f(U_cU_c^\dagger)), \qquad f(A) :=\frac{\kappa_1}{2}A+\frac{\kappa_2}{4}A^2+\cdots+\frac{\kappa_m}{2m}A^m.
\end{equation*}
Then, one has emergent dynamics (see Theorem \ref{T4.2}):
\[ \lim_{t \to \infty} \|U_j(t) - U_j^{\infty} \|_F = 0, \qquad \lim_{t \to \infty} \frac{d}{dt} {\mathcal V}_{poly}(U) = 0, \qquad  \lim_{t\rightarrow\infty} \| \dot{U}_j \|_F 
=0, \quad   j = 1, \cdots, N. \]

\vspace{0.5cm}

The rest of this paper is organized as follows. In Section \ref{sec:2}, we briefly discuss the Lohe matrix model and its basic properties, and we also present a gradient flow formulation, and we briefly review basic a priori estimates to be used crucially for a later use. In Section \ref{sec:3}, we consider a monomial potential function and as a gradient flow approach, we derive a generalized Lohe matrix model with a higher-order couplings, and study several emergent estimates and provide several sufficient frameworks leading to the emergent dynamics. In Section \ref{sec:4}, we consider a general case with a polynomial potential and derive a generalized Lohe matrix model with higher-order couplings. Using Barbalat's lemma, we derive an emergent dynamics of the proposed model. In Section \ref{sec:5}, we derive a Gronwall type differential inequality for a ensemble diameter. This yields an exponential decay estimate of ensemble diameter. Finally, Section \ref{sec:6} is devoted to a brief summary of our main results and some unresolved issues for a future work. 

\bigskip

\noindent {\bf Notation}: Let $\bbu(d)$ be a unitary group manifold consisting of unitary $d \times d$ matrix $U^{\dagger} U = U U^{\dagger} = I_d$ and for two unitary matrices $A, B \in \bbu(d)$, we introduce a Frobenius inner product  $\langle \cdot, \cdot \rangle_F$ and its induced norm $\| \cdot \|_F$:
\[ \langle A, B \rangle_F := \mbox{tr}(A^\dagger B), \quad \|A \|_F := \sqrt{\langle A, A \rangle_F}. \]

\section{Preliminaries} \label{sec:2}
\setcounter{equation}{0} 
In this section, we review the Lohe matrix model \cite{Lo-1, Lo-2} on the unitary group $\bbu(d)$ and review the basic properties of the Lohe matrix model such as conservation laws and gradient flow formulation.
\subsection{The Lohe matrix model} \label{sec:2.1}
Let $U_j$ be a $d\times d$ unitary matrix whose dynamics is governed by the first-order continuous-time dynamical system:
\begin{equation} \label{B-1}
{\mathrm i}\dot{U}_jU_j^\dagger=H_j+ \frac{{\mathrm i}\kappa}{2N}\sum_{k=1}^N\left( U_kU_j^\dagger-U_jU_k^\dagger\right), \quad j = 1, \cdots, N,
\end{equation}
where $\kappa$ is a nonnegative coupling strength, $U_j^\dagger$ denotes the hermitian conjugate of the matrix $U_j$, and $H_j$ is the Hermitian matrix with the property $H_j^\dagger = H_j$. This property results in the following relation: 
\begin{equation*} \label{B-2}
\langle U_j, -{\mathrm i} H_j U_j \rangle_F + \langle -{\mathrm i} H_j U_j, U_j \rangle_F =0, \quad j = 1, \cdots, N,
\end{equation*}
where $\langle \cdot, \cdot \rangle_F$ is the Frobenius inner product on ${\mathbb U}(d)$:
\[ \langle A, B \rangle_F := \mbox{tr}(A^\dagger B), \quad A, B \in {\mathbb U}(d). \]
Below, we will see that the quadratic quantity $U_j^\dagger U_j$ is a conserved quantity (see Proposition \ref{P2.1}).  

For the case $U_j U_j^\dagger = I_d$, system \eqref{B-1} can be rewritten as 
\begin{equation} \label{B-3}
\dot{U}_j = -{\mathrm i}\ H_j U_j +\frac{\kappa}{2N}\sum_{k=1}^N\left( U_k  -U_jU_k^\dagger U_j \right),  \quad j = 1, \cdots, N.
\end{equation}
Moreover, system \eqref{B-3} can be further simplified as a mean-field form using the average quantity $U_c:={1\over{N}}\sum_{k=1}^N U_k$ to rewrite system \eqref{B-3} as 
\begin{equation*} \label{B-4}
\dot{U}_j= -{\mathrm i} H_j U_j +  \frac{\kappa}{2} \left( U_c-U_j U_c^\dagger U_j\right).
\end{equation*}
Next, we list several key properties of \eqref{B-1} as follows.
\begin{proposition} \label{P2.1}
\emph{\cite{Lo-1, Lo-2}} 
Let $\{U_j\}$ be a global smooth solution to \eqref{B-1} with the initial data $\{U_j^{in} \}$. Then, the following assertions hold.
\begin{enumerate}
\item
(Conservation of amplitude): The quadratic quantity $U_j U_j^\dagger$ is conserved along the Lohe matrix flow \eqref{B-1}: 
\[  U_j(t) U_j^\dagger(t) = U_j^{in} U_j^{in \dagger}, \quad t > 0. \]
\item
(Unitary invariance): Let ${\tilde U}_j$ be a transformed state by the relation:
\[ {\tilde U}_j := U_j L, \quad L \in {\mathbb U}(d). \]
Then, the transformed state ${\tilde U}_j$ satisfies
\[ {\mathrm i} \dot{{\tilde U}}_j {\tilde U}_j^\dagger = H_j + \frac{{\mathrm i} \kappa}{2N}\sum_{k = 1}^{N} \left( {\tilde U}_k {\tilde U}_j^\dagger - {\tilde U}_j {\tilde U}_k^\dagger \right), \quad t > 0. \]
\end{enumerate}
\end{proposition}

\vspace{0.5cm}

As in the Kuramoto model, system \eqref{B-1} admits a ``{\it  solution splitting property}"  for the identical Hamiltonian case:
\[  H_j = H, \quad j = 1, \cdots, N. \]
In this case, system \eqref{B-1} becomes
\begin{equation} \label{B-5}
\displaystyle \dot{U}_j  = -{\mathrm i} H U_j + \frac{\kappa}{2N}\sum_{k = 1}^{N} (U_k  - U_j U_k^\dagger U_j ), \quad j = 1, \cdots, N.
\end{equation}
Let ${\mathcal R}(t)$ and ${\mathcal L}(t)$ be the two solution operators corresponding to the following two subsystems, respectively:
\begin{equation*} \label{S-L}
\displaystyle \dot{V}_j  = -{\mathrm i} H V_j,  \quad \dot{W}_j  =  \frac{\kappa}{2N}\sum_{k = 1}^{N} (W_k  - W_j W_k^\dagger W_j ).
\end{equation*}
Next, we introduce solutions operators associated with the above two systems:
\[ 
{\mathcal R}(t) {\mathcal V}^{in}  := (e^{-{\mathrm i} H t} V_1^{in}, \cdots, e^{-{\mathrm i} H t} V_N^{in}), \quad  {\mathcal W}(t) L^{in}  := (W_1(t), \cdots, W_N(t)), \quad t \geq 0. 
\]
\begin{proposition} \label{P2.2}
\emph{\cite{H-R}}
Let ${\mathcal S}(t)$ be a solution operator to \eqref{B-5}. Then, one has
\[ 
{\mathcal S}(t) =  {\mathcal R}(t) \circ {\mathcal L}(t),\quad t \geq 0.
\]    
\end{proposition}
It follows from Proposition \ref{P2.2} that it suffices to assume $H = 0$ for a homogeneous ensemble in what follows.

\subsection{A gradient flow formulation}  \label{sec:2.2} Consider the Lohe matrix model with $H \equiv 0$:
\begin{equation} \label{B-5-1}
 \dot{U}_j =   \frac{\kappa}{2N}\sum_{k=1}^N \left(U_k -U_j U_k^\dagger  U_j  \right), \quad j = 1, \cdots, N.
\end{equation} 
In \cite{H-K-R2}, the authors introduced an order parameter $R$ and the corresponding potential ${\mathcal V}_1$ for \eqref{B-5-1} with $H_i = 0$:
\begin{equation} \label{B-6}
  R^2 :=\frac{1}{N^2}\sum_{i,j=1}^N \mbox{tr}\left(U^{\dagger}_i U_j \right) = \mbox{tr}(U_c^+ U_c) \quad \mbox{and} \quad  \mathcal{V}_1:=-\frac{\kappa N}{2} R^2 = -\frac{\kappa N}{2} \|U_c \|_F^2,
\end{equation}
Then, it is easy to see that $R^2$ is analytic and 
\begin{equation} \label{B-6-1}
 R =  \|U_c \|_F \leq \frac{1}{N} \sum_{j=1}^{N} \|U_j \|_F  = \sqrt{d}.
\end{equation}

Note that the potential ${\mathcal V}_1$ is an analytic function of states $U_j$'s, and the Riemannian metric on $\bbu(d)$ is induced by the natural inclusion $\bbu(d)\hookrightarrow M_{d,d}(\bbc)$.
\begin{proposition} \label{P2.3}
\emph{\cite{H-K-R2}}
The Lohe matrix model \eqref{B-5-1} with $H = 0$ is a gradient flow with an analytical potential ${\mathcal V}_1$ in \eqref{B-6}:
\[ {\dot U}_i = -\left.\frac{\partial \mathcal{V}_1}{\partial U_i}\right|_{T_{U_i}\bbu(d)}, \quad t > 0, \quad i = 1, \cdots, N.   \]
\end{proposition}
\begin{proof}
Although a detailed proof can be found in \cite{H-K-R2}, we briefly sketch its proof here for self-containedness. \newline

\noindent $\bullet$~Step A (Expression of the potential in terms of components of $U_i$):~Note that the function $\mathcal{V}$ has an obvious polynomial extension to all of $M_{d,d}(\bbc)^N=\bbc^{2d^2N}$ viewed as a real analytic manifold. Since each variable $U_i$ is in $M_{d,d}(\bbc)=\bbc^{2d^2}$, the partial derivatives of a matrix can be calculated by the partial derivatives of each real and imaginary component of $U_i$ on $\bbr^{2d^2}$. Let $u^{kl}_i=a^{kl}_i+{\mathrm i}b^{kl}_i$ be the $(k,l)$-element of matrix $U_i$, where $a_i^{kl}$ and $b_i^{kl}$ are real numbers. First, we use 
\[\mbox{tr}\left(U_i U_j^\dagger\right)  = \sum_{k,l = 1}^{d} u_i^{kl} {\bar u_j^{kl}}=\sum_{k,l = 1}^{d}\left[ (a_i^{kl}a_j^{kl}+ b_i^{kl}b_j^{kl})+{\mathrm i}(a_j^{kl}b_i^{kl}-a_i^{kl}b_j^{kl})\right] \]
to see
\[ \label{B-6-2}
{\mathcal V}_1 = -\frac{\kappa N}{2}\|U_c\|_F^2=-\frac{\kappa}{2N}\sum_{i, j=1}^N\mathrm{tr}(U_iU_j^\dagger) =  -\frac{\kappa}{2N}  \sum_{i,j=1}^N   \sum_{k,l = 1}^{d} \left[ a_i^{kl}a_j^{kl}+ b_i^{kl}b_j^{kl}\right],
\]
where we cancel the imaginary term by symmetry of the indices $i$, $j$. \newline

\noindent $\bullet$~Step B: We derive
\begin{equation*} \label{B-7}
 \left.\frac{\partial \mathcal{V}_1}{\partial U_i}\right|_{T_{U_i}M_{d,d}(\bbc)} = -\frac{\kappa}{N}\sum_{j=1}^N U_j = -\kappa U_c.    
\end{equation*} 
By direct calculation with \eqref{B-6}, one has 
\[
\frac{\partial \mathcal{V}_1}{\partial a^{kl}_i}=-\frac{\kappa}{N}\sum_{j=1}^N a^{kl}_j, \quad \frac{\partial \mathcal{V}_1}{\partial b^{kl}_i}=-\frac{\kappa}{N}\sum_{j=1}^N b^{kl}_j,
\]
and thus reverting back to the coordinates of $M_{d,d}(\bbc)^N=\bbr^{2d^2N}$, we have
\[
\left.\frac{\partial \mathcal{V}_1}{\partial U_i}\right|_{T_{U_i}M_{d,d}(\bbc)} =\sum_{k,l=1}^d\left(\frac{\partial \mathcal{V}_1}{\partial a^{kl}_i}+{\mathrm i}\frac{\partial \mathcal{V}_1}{\partial b^{kl}_i}\right)E^{kl} =-\frac{\kappa}{N}\sum_{j=1}^N\sum_{k,l=1}^d u^{kl}_jE^{kl} =-\frac{\kappa}{N}\sum_{j=1}^N U_j,
\]
where $E^{kl}$ denotes the $d\times d$ matrix whose $(k,l)$-coordinate is $1$ and the other coordinates are $0$. \newline

\noindent $\bullet$~Step C: We set 
\[  \mathfrak{u}(d)=T_{I_d}\bbu(d) = \{ X \in M_{d,d}(\bbc) ~|~ X + X^\dagger = 0 \}, \]
and define an orthogonal projection:
\[ \pi:T_{I_d}M_{d,d}(\bbc)\rightarrow\mathfrak{u}(d) \quad \mbox{by}~~ A \mapsto \frac{1}{2}(A-A^\dagger). \]
Since ${T_{U_i}\bbu(d)}$ is the right translate $\mathfrak{u}(d)U_i$ of $\mathfrak{u}(d)$, we can see that the orthogonal projection $\pi_{U_i}:T_{U_i}M_{d,d}(\bbc) \rightarrow {T_{U_i}\bbu(d)}$ is given by $AU_i \mapsto \pi(A)U_i=\frac{1}{2}(A-A^\dagger)U_i$ for an element $AU_i \in T_{U_i}M_{d,d}(\bbc)$. Hence we may calculate
\begin{align*}
\left.\frac{\partial \mathcal{V}_1}{\partial U_i}\right|_{T_{U_i}\bbu(d)} &= \pi_{U_i}\left( \left.\frac{\partial \mathcal{V}_1}{\partial U_i}\right|_{T_{U_i}M_{d,d}(\bbc)}  \right) = \pi\left( \left.\frac{\partial \mathcal{V}_1}{\partial U_i}\right|_{T_{U_i}M_{d,d}(\bbc)} U_i^\dagger \right)U_i \\
&=\pi\left(-\frac{\kappa}{N}\sum_{j=1}^N U_jU_i^\dagger\right)U_i = -\frac{\kappa}{2N}\sum_{j=1}^N \left(U_j U_i^\dagger -U_i U_j^\dagger\right) U_i \\
&= -\frac{\kappa}{2} (U_c - U_i U^{\dagger}_c U_i).
\end{align*}
\end{proof}
As a direct application of Proposition \ref{P2.3}, we have the convergence of the flow $e^{-{\mathrm i} Ht}U_i$ as $t \to \infty$.
\begin{corollary} \label{C2.1}
Let $U_i = U_i(t)$ be a global solution to the Cauchy problem \eqref{B-5-1}. Then, the flow $U_i$ converges for any initial configuration $\{U_i^{in}\}$.
\end{corollary}
\begin{proof} We use a gradient flow formulation in Proposition \ref{P2.3} and a standard argument in \cite{H-K-R2} to derive the convergence of the flow.
\end{proof}
\begin{lemma}  \label{L2.1}
\emph{(Babalat's lemma \cite{Ba})}
Suppose that a real-valued function $f: [0, \infty) \to \bbr$ is continuously differentiable, and $\lim_{t \to \infty} f(t) = \alpha \in \bbr$. If $f^{\prime}$ is uniformly continuous, then 
\[ \lim_{t \to \infty} f^{\prime}(t)  = 0. \]
\end{lemma}
\begin{proposition} \label{P2.4}
Let $\{ U_j \}$ be a global solution to the Cauchy problem \eqref{B-5}. Then,  the potential ${\mathcal V}_1$ in \eqref{B-6} satisfies
 \[ \frac{d{\mathcal V}_1}{dt} = -\sum_{i=1}^N \|\dot{U}_i\|_F^2, \quad  \exists~\lim_{t \to \infty} {\mathcal V}_1(U(t)) \quad \mbox{and} \quad \sup_{0 \leq t < \infty} \Big| \frac{d^2}{dt^2} {\mathcal V}_1(U(t)) \Big| < \infty.
\]
\end{proposition}
\begin{proof}
Note that $U_j$ and $U_c$ satisfy 
\begin{equation} \label{B-8}
\dot{U}_j= \kappa(U_c-U_j U_c^\dagger U_j), \quad  \dot{U}_c=\frac{\kappa }{N}\sum_{j=1}^N (U_c-U_j U_c^\dagger U_j), \quad j  = 1, \cdots, N.
\end{equation}

\noindent $\bullet$ (Estimate of the first and second estimates): We use the above relations \eqref{B-8} to get 
\begin{equation} \label{B-9}
\frac{dR^2}{dt} =\frac{1}{\kappa N}\sum_{j=1}^N \|\dot{U}_j \|_F^2 \geq 0.  
\end{equation}
This yields the first estimate:
\[  \frac{d \mathcal{V}_1}{dt} :=- \sum_{j=1}^N \|\dot{U}_j \|_F^2 \leq 0. \]
Since ${\mathcal V}_1$ is non-decreasing and bounded below by $-\frac{1}{2}\kappa N d$ (see \eqref{B-6-1}), 
\[ \exists~\lim_{t \to \infty} {\mathcal V}_1(U(\cdot)). \]

\vspace{0.2cm}

\noindent $\bullet$ (Estimate of the third estimate): In the sequel, we will derive
\begin{equation*} \label{B-10}
 \|\dot{U}_i\|_F \leq  \frac{\kappa}{2}\sqrt{d}(1+d), \quad \|\dot{U}_c\|_F \leq \frac{\kappa}{2}\sqrt{d}(1+d), \quad 
\left|\frac{d}{dt}\|\dot{U}_j \|_F^2\right| \leq \frac{\kappa^3}{2}d(1+d)(1+2d).
\end{equation*}
For the first estimate, we use \eqref{B-6-1} and $\eqref{B-8}_1$ to get 
\[  \|\dot{U}_j \|_F \leq \frac{\kappa }{2}\Big( \|U_c \|_F + \| U_j U_c^\dagger U_j \|_F  \Big) \leq \frac{\kappa }{2}\Big( \|U_c \|_F + \| U_j \|_F \cdot \| U_c^\dagger \|_F \cdot  \|U_j \|_F  \Big)  \leq \frac{\kappa }{2} \sqrt{d}(1 + d). \] 
This also implies 
\[ \|\dot{U}_c\|_F =\left\|\frac{1}{N}\sum_{j=1}^N \dot{U}_j \right\|_F \le \frac{1}{N}\sum_{j=1}^N \|\dot{U}_j \|_F \le \frac{\kappa}{2}\sqrt{d}(1+d). \]
For the third estimate, we use $\eqref{B-8}_1$ to obtain 
\begin{align*}
\begin{aligned}
\frac{d}{dt}\|\dot{U}_j \|_F^2=&\frac{d}{dt}\frac{\kappa^2}{2} \mbox{Re}~ \mbox{tr} (U_cU_c^\dagger-U_j U_c^\dagger U_j U_c^\dagger)\\
=&\frac{\kappa^2}{2} \mbox{Re}~ \mbox{tr} (\dot{U}_cU_c^\dagger+U_c\dot{U}_c^\dagger-\dot{U}_j U_c^\dagger U_j U_c^\dagger-U_j \dot{U}_c^\dagger U_j U_c^\dagger-U_i U_c^\dagger \dot{U}_j U_c^\dagger-U_j U_c^\dagger U_j \dot{U}_c^\dagger)\\
=&\kappa^2 \mbox{Re}~ \mbox{tr} (\dot{U}_cU_c^\dagger-\dot{U}_j U_c^\dagger U_j U_c^\dagger-U_j \dot{U}_c^\dagger U_j U_c^\dagger).
\end{aligned}
\end{align*}
This yields
\begin{align*}
\begin{aligned}
\left|\frac{d}{dt}\|\dot{U}_j \|_F^2\right| &\le \kappa^2(\|\dot{U}_c \|_F \|U_c^\dagger\|+\|\dot{U}_j U_c^\dagger\|_F \| U_j U_c^\dagger\|_F +\|U_j \dot{U}_c^\dagger\|_F \| U_j U_c^\dagger\|_F) \\
&\le \frac{\kappa^3}{2}d(1+d)(1+2d).
\end{aligned}
\end{align*}
Finally, \eqref{B-9} yields
\[   \frac{d^2R^2}{dt^2} =\frac{2}{\kappa N}\sum_{j=1}^N \frac{d}{dt} \|\dot{U}_j \|_F^2 \leq \kappa^2d(1+d)(1+2d).   \]
This yields
\[  \sup_{0 \leq t < \infty} \Big| \frac{d^2{\mathcal V}_1}{dt^2} \Big| \leq N \kappa^3 d(1+d)(1+2d). \]
\end{proof}
Finally, we combine Lemma \ref{L2.1} and Proposition \ref{P2.4} to get a desired result. 
\begin{corollary} \label{C2.2}
Let $U_j = U_j(t)$ be a global solution to the Cauchy problem \eqref{B-5}. Then, one has 
\[
\lim_{t\rightarrow\infty} \dot{U_j}(t)=0 \text{ for all }~ j = 1, \cdots, N.
\]
\end{corollary}

%\begin{remark}
%As we discussed on $\bbs^d$, we can consider the existence of potential functions for nonidentical oscillators. In $\bbu(d)$, a possible potential function for $\kappa=0$ is 
%\[ {\mathcal V}(U) = {\mathrm i} \sum_{j=1}^N \mbox{tr}(U_j^\dagger H_j U_j).\]
% However, this is a constant function on $\bbu(d)$ since all $U_j$ are unitary. We can rigorously prove the nonexistence of potential functions on $\bbu(d)$ by showing periodic solutions.
%\begin{example}\label{R5.2}
%For the case of $\kappa = 0$ and $H_1 = h_1 I_d$ with a nonzero real scalar $h_1$, the oscillator $U_1$ follows the equation $\dot{U_1} = -{\mathrm i}h_1 U_1$. Therefore, we have a periodic solution $U_1(t) = \exp(-{\mathrm i}h_1 t)U_1^{in}$.
%\end{example}
%Since $\bbu(d)$ is not a simply connected manifold, we have the following potential function on the covering space $\bbr \times S\bbu(d)$ similar to the Kuramoto model. The universal covering space of $\bbu(d)$ is $\bbr \times S\bbu(d)$. A Lohe flow on $\bbu(d)$ with nonidentical oscillators can have a potential function if and only if the natural frequencies $H_i$ are all scalars. We will see this and analyze a Lohe flow on $\bbr \times S\bbu(d)$ in a later subsection.
%\end{remark}

\subsection{Previous results} \label{sec:2.3}
In this subsection, we brief review two closely related results from \cite{D, H-K-R2} that deal with Lohe type aggregation model on matrix Lie groups.  The original Lohe matrix model was introduced as an aggregation model on the unitary group. After Lohe's  work \cite{Lo-1, Lo-2}, it was further extended to the Lohe group (whose definition is defined below) in \cite{H-K-R2}. Next, we provide a concept of the Lohe group in the following definition.
\begin{definition} \label{D2.1}
\emph{\cite{H-K-R2}}
Let $G$ be a subgroup of  the linear group $GL(d, k)$ with  $k=\bbr$ or $k=\bbc$. Then, we say $G$ is a Lohe group if the following two conditions hold.
\begin{enumerate}
\item
$G$ is a matrix Lie group, which is a closed subgroup of $GL(d, k)$:
\[
 G \leq GL(d, k). 
\]
\item
$G$ satisfies
\[
X-X^{-1}\in \mathfrak{g}\quad \text{ for all } X\in G.
\]
Here, $\mathfrak{g}$ is the Lie algebra associated to $G$, i.e., the tangent space $T_{I}G$ of $G$ at the identity matrix $I$.
\end{enumerate}
\end{definition}
\begin{remark}
The following matrix groups are Lohe groups:
\[  GL(k,d),~O_d(k),~O_{p,q}(k),~U_d,~SP_{n}(k),~SO_d(k),~SO_{p,q}(k),~ SP(n),~ SL(2,k),~ SU(2). \]
\end{remark}

\vspace{0.5cm}

Let $G$ and $\mathfrak{g}$ be a Lohe group and its associated Lie algebra, respectively. Then, a generalized Lohe matrix model \cite{H-K-R2} on $G$ reads as follows.
\begin{equation} \label{GL}
\begin{cases}
\displaystyle \dot{X}_j X_j^{-1}= \Omega_j + \frac{\kappa}{2N}\sum_{k=1}^{N} \Big ( X_k X_j^{-1} - X_j X_k^{-1} \Big), \quad t > 0, \\
\displaystyle X_j(0)=X^{in}_j,~~ j=1,\cdots,N,
\end{cases}
\end{equation}
where $\Omega_i\in \mathfrak{g}$.\newline

 For an ensemble $\{X_j \}$, we set 
\[
{\mathcal D}(X) := \max_{1 \leq i,j \leq N} \| X_i - X_j \|_F. 
\]

\begin{theorem}\label{T2.1}
\emph{\cite{H-K-R2}}
Suppose that $H_i$, the coupling strength and the initial data $\{X^{in} \}$ satisfy
\[ H_i =0, \quad  i = 1, \cdots, N, \quad  \kappa > 0 \quad \mbox{and} \quad {\mathcal D}(X^{in})<1. \]
Then, there exists a smooth global solution $\{ X_j \}$ such that 
\[ \frac{{\mathcal D}(X^{in})}{(1+ {\mathcal D}(X^{in}))e^{\kappa t}- {\mathcal D}(X^{in})}\leq {\mathcal D}(X(t))\leq \frac{ {\mathcal D}(X^{in})}{(1-{\mathcal D}(X^{in}))e^{\kappa t}+ {\mathcal D}(X^{in})}, \quad t \geq 0.
\]
\end{theorem}

\vspace{0.5cm}

In \cite{D}, Deville further introduced the generalized model \eqref{GL} by introducing a polynomial coupling, namely ``{\it quantum Kuramoto model}"  on the Lohe group $G$ associated with Lie algebra $\mathfrak{g}$ introduced in Definition \ref{D2.1}.  Let $\Gamma$ be an undirected weighted graph with $N$ vertices with edge weights $\gamma_{ij} \geq 0$ with $\gamma_{ij} = \gamma_{ji}$. To be more specific, let $f$ be a real analytic function. Then, the quantum Kuramoto model reads as 
\begin{equation} \label{QK}
{\dot X}_j  X_j^{-1} = \Omega_j + \frac{1}{2} \sum_{k=1}^{N} \gamma_{jk} \Big( f(X_k X_j^{-1}) - f(X_j X_k^{-1})  \Big),
\end{equation}
where $\Omega_j \in {\mathfrak g}$. \newline

Note that for $f(x) = x$  and $\gamma_{ij} = \frac{2}{N}$, the quantum Kuramoto model \eqref{QK} becomes the generalized Lohe matrix model \eqref{GL}. In the aforementioned work, Deville studied sync and near sync solutions and investigated the stability of these solutions and twist solutions. Other than a linear function $f$, system \eqref{QK} cannot be rewritten as a mean-field form. For example, $f(x) = x^2$ and $\gamma_{jk} = \frac{2}{N}$, system \eqref{QK} becomes 
\begin{equation*} \label{QK-1}
{\dot X}_j  X_j^{-1} = \Omega_j + \frac{1}{N} \sum_{k=1}^{N} \Big( (X_k X_j^{-1})^2 - (X_j X_k^{-1})^2  \Big).
\end{equation*}
This is clearly different from our proposed model \eqref{C-1}: 
\[  \dot{U}_j  U_j^{\dagger} = -{\mathrm i} H_j +  \frac{\kappa}{N^{3}} \sum_{k_1, k_2, k_3 = 1}^{N} \Big(  U_{k_1} U^\dagger_{k_2} U_{k_3} U_j^{\dagger}  - U_j U_{k_3}^{\dagger} U_{k_2} U_{k_1}^{\dagger}  \Big). \]

\subsection{Elementary estimates} \label{sec:2.4}
In this subsection, we present several elementary estimates to be used for later sections. \newline 

Let $U_i$ and $U_j$ be the unitary matrices. Then we have following identities:
\[
\|U_i-U_j\|_F^2=2d-\mathrm{tr}(U_iU_j^*-U_jU_i^*),\qquad \|U_c\|_F^2=d-\frac{1}{2N^2}\sum_{i, j}\|U_i-U_j\|_F^2.
\]

\begin{lemma}\label{L2.2}
Let $A$ and $B$ be the matrices with proper size. Then we have the following inequality:
\[
\|AB\|_F\leq\|A\|_{op}\cdot\|B\|_F,
\]
where $\|\cdot\|_{op}$ is an operator norm.
\end{lemma}

\begin{proof}
Let 
\[
B=\left[b_1\ \vdots \ b_2\ \vdots\ \cdots\ \vdots \ b_n\ \right],
\]
where $b_\alpha$ is a vector. From the direct calculation, we have
\begin{align*}
\|AB\|_F^2&=\sum_{\alpha,\beta}|[AB]_{\alpha\beta}|^2=\sum_{\alpha}\|Ab_\alpha\|^2\leq \sum_{\alpha}\left(\|A\|_{op}\cdot\|b_\alpha\|\right)^2 \\
&=\|A\|_{op}^2\cdot\sum_{\alpha}\|b_\alpha\|^2 = \|A\|_{op}^2\cdot\|B\|_{F}^2.
\end{align*}
\end{proof}

\begin{remark}
If $U$ is a unitary matrix, then 
\[
\|U\|_{op}=1.
\]
\end{remark}

\begin{lemma}\label{L2.3}
Let $\{ U_i \}_{i=1}^{N}$ be an ensemble of unitary matrices in $\bbu(d)$. Then we have
\[
\|U_c\|_{op}\leq 1.
\]
Equality holds if and only if there exists $v\in\mathbb{C}^d$ with $\|v\|\neq0$ such that
\[
U_1v=U_2v=\cdots =U_Nv.
\]
\end{lemma}

\begin{proof}
By direct calculation, we have
\[
\|U_cv\|_F\leq \frac{1}{N}\sum_{k=1}^N\|U_kv\|_F\leq\frac{1}{N}\sum_{k=1}^N\|U_k\|_{op}\|v\|=\|v\|.
\]
So we have
\[
\|U_c\|_{op}\leq 1.
\]
We can also easily show the equality condition.
\end{proof}

\begin{lemma}\label{L2.4}
Let $A$ and $B$ be $d \times d$ matrices. Then, one has 
\[
|\mathrm{tr}(A)|\leq\sqrt{d}\|A\|_F \quad \mbox{and} \quad \|AB\|_F\leq\|A\|_F\cdot\|B\|_F. 
\]
\end{lemma}
\begin{proof}
\noindent (i)~By definition of a trace, one has 
\begin{align*}
|\mathrm{tr}(A)|^2=\left|\sum_{\alpha}[A]_{\alpha\alpha}\right|^2\leq\sum_{\alpha}|[A]_{\alpha\alpha}|^2\cdot\sum_\alpha 1^2\leq d\|A\|_F^2.
\end{align*}
\noindent (ii)~By direct calculation, one has
\begin{align*}
\|AB\|_F^2&=\sum_{\alpha,\beta}|[AB]_{\alpha\beta}|^2=\sum_{\alpha,\beta}|\sum_\gamma[A]_{\alpha\gamma}[B]_{\gamma\beta}|^2\\
&\leq \sum_{\alpha,\beta}\left(\sum_\gamma|[A]_{\alpha\gamma}|^2\right)\cdot\left(\sum_\gamma|[B]_{\gamma\beta}|^2\right)=\|A\|_F^2\cdot\|B\|_F^2.
\end{align*}
Thus, we have
\[
\|AB\|_F\leq\|A\|_F\cdot\|B\|_F.
\]
\end{proof}

\section{The Lohe matrix model with a monomial interaction} \label{sec:3}
\setcounter{equation}{0} 
In this section, we present a generalized Lohe matrix model with a monomial higher-order coupling via a gradient flow approach. \newline

More precisely, we will derive a generalized Lohe matrix model with higher-order couplings: for $t > 0$, 
\begin{equation}
\begin{cases} \label{C-0}
\displaystyle {\mathrm i} \dot{U}_j  U_j^{\dagger} = H +  \frac{{\mathrm i} \kappa}{2N^{2m-1}} \\
\displaystyle \hspace{0.5cm} \times \sum_{k_1, \cdots, k_{2m-1} = 1}^{N}  \Big(  U_{k_1} U_{k_2}^{\dagger} \cdots U_{k_{2m-2}}^{\dagger} U_{k_{2m-1}} U_j^{\dagger}  - U_j  U_{k_{2m-1}}^{\dagger} U_{k_{2m-2}} \cdots U_{k_2} U_{k_1}^{\dagger}  \Big),\\
\displaystyle U_j \Big|_{t = 0+} = U_j^{in},
 \end{cases}
 \end{equation} 
 where $H$ is a Hermitian matrix with $H^\dagger = H$. \newline
 
Note that system \eqref{C-0} can be rewritten as a mean-field form using a mean-field quantity $U_c$:
 \begin{equation} \label{N-0}
  {\mathrm i} \dot{U}_j  U_j^{\dagger} = H + \frac{ {\mathrm i} \kappa }{2} 
  \Big(\underbrace{U_cU_c^\dagger U_c\cdots U_c^\dagger U_c}_{2m-1} U_j^{\dagger} -U_j\underbrace{U_c^\dagger U_cU_c^\dagger \cdots U_cU_c^\dagger}_{2m-1} \Big).
 \end{equation}
 Since the R.H.S. of \eqref{C-0} is self-adjoint, and this yields the conservation of quadratic quantities $U_j U_j^{\dagger}$.
 \begin{lemma} \label{L3.1}
Let $\{U_i\}$ be a global smooth solution of  system \eqref{C-0}. Then  one has 
\[ \frac{d}{dt} (U_j U^{\dagger}_j) = 0, \quad t > 0,~~ j = 1, \cdots, N. \] 
\end{lemma}
\begin{proof} We set 
\begin{align*}
\begin{aligned}
&{\mathcal C}_j(U_1, \cdots, U_N) \\
& \hspace{0.5cm} := \frac{{\mathrm i} \kappa}{2N^{2m-1}} \sum_{k_1, \cdots, k_{2m-1} = 1}^{N}  \Big(  U_{k_1} U_{k_2}^{\dagger} \cdots U_{k_{2m-2}}^{\dagger} U_{k_{2m-1}} U_j^{\dagger}  - U_j  U_{k_{2m-1}}^{\dagger} U_{k_{2m-2}} \cdots U_{k_2} U_{k_1}^{\dagger}  \Big).
\end{aligned}
\end{align*}
Then, it is easy to see that 
\begin{equation} \label{N-1}
{\mathcal C}_j(U_1, \cdots, U_N)^{\dagger} = {\mathcal C}_j(U_1, \cdots, U_N). 
\end{equation} 
Now, we return to system \eqref{C-0}:
\begin{equation}\label{N-2}
\dot{U}_j  U_j^{\dagger}  = -{\mathrm i} \Big( H + {\mathcal C}_j(U_1, \cdots, U_N) \Big). 
\end{equation}
We take a hermitian conjugate of \eqref{N-2} and use \eqref{N-1} to get 
\begin{equation}\label{N-3}
U_j  {\dot U}^{\dagger}_j  ={\mathrm i}  \Big( H + {\mathcal C}_j(U_1, \cdots, U_N) \Big). 
\end{equation}
Finally, we add \eqref{N-2} and \eqref{N-3} to get 
\[ \frac{d}{dt} (U_j U_j^{\dagger}) = 0. \]
\end{proof}
From now on, throughout the paper, we assume 
\[ U_j U_j^{\dagger} = U_j^{\dagger} U_j = I_d, \quad j = 1, \cdots, N. \]
and consider emergent dynamics of the following Cauchy problem:
\begin{equation}
\begin{cases} \label{N-4}
\displaystyle \dot{U}_j = - {\mathrm i} H U_j  +  \frac{\kappa}{2} 
  \Big(\underbrace{U_cU_c^\dagger U_c\cdots U_c^\dagger U_c}_{2m-1} -U_j\underbrace{U_c^\dagger U_cU_c^\dagger \cdots U_cU_c^\dagger}_{2m-1} U_j \Big),~~t > 0, \\
\displaystyle U_j \Big|_{t = 0+} = U_j^{in},\quad j = 1, \cdots, N.
 \end{cases}
 \end{equation} 
 Now, we consider the corresponding nonlinear subsystem:
\begin{equation}
\begin{cases} \label{N-5}
\displaystyle  \dot{L}_j  =  \frac{\kappa}{2} 
  \Big(\underbrace{L_c L_c^\dagger L_c\cdots L_c^\dagger L_c}_{2m-1} - L_j\underbrace{L_c^\dagger L_c L_c^\dagger \cdots L_c L_c^\dagger}_{2m-1} L_j \Big), ~~t > 0,\\
\displaystyle L_j \Big|_{t = 0+} = U_j^{in}, \quad j  = 1, \cdots, N.
 \end{cases}
 \end{equation} 
 \begin{proposition}
 \emph{(Solution splitting property)}
 Let $\{U_j \}$ and $\{L_j \}$ be two solutions to systems \eqref{N-4} and \eqref{N-5}, respectively. Then one has 
 \begin{equation} \label{N-6}
  U_j(t) = e^{-{\mathrm i} H t } \circ L_j(t), \quad j = 1, \cdots, N. 
  \end{equation}
 \end{proposition}
 \begin{proof}
 We substitute \eqref{N-6} into \eqref{N-5} and use the relations
 \[  {\dot U}_j =  -{\mathrm i} H e^{-{\mathrm i} H t } L_j + e^{-{\mathrm i} H t }  {\dot L}_j, \quad U_c  U_c^{\dagger} = L_c L_c^{\dagger}, \quad  U_j U_c^{\dagger} = L_j L_c^{\dagger}   \]
 to see that $L_j$ satisfies system \eqref{N-5}.
  \end{proof}
Fom now on, we assume $H \equiv 0$. In what follows, we will derive system \eqref{C-0} using a gradient flow formulation with a monomial potential $ {\mathcal V}_m(U)$: for $m \geq 1$, we set 
\begin{equation*} \label{C-0-1}
 {\mathcal V}_m(U) :=-\frac{\kappa N}{2m}\mathrm{tr}((U_cU_c^\dagger)^m)=-\frac{\kappa N}{2m}\mathrm{tr}(\underbrace{U_cU_c^\dagger \cdots U_cU_c^\dagger}_{2m}).
\end{equation*}
Note that ${\mathcal V}_m$ is analytic and bounded:
\begin{equation} \label{C-4-2-2}
 | {\mathcal V}_m(U)| \leq \frac{\kappa}{2m}|\mathrm{tr}(\underbrace{U_cU_c^\dagger \cdots U_cU_c^\dagger}_{2m}) | \leq  \frac{\kappa}{2m} \|U_c \|_F^{2m} \leq  \frac{\kappa}{2m} d^{m}. 
 \end{equation}
 
\vspace{0.5cm}

In next two subsections, we consider the cases:
\[ \mbox{Either}~~m = 2 \quad \mbox{or} \quad m \geq 3. \]

\subsection{Case with $m = 2$} \label{sec:3.1} Consider the Cauchy problem to the following system:
\begin{align}
\begin{aligned} \label{C-1}
& \dot{U}_j = \frac{\kappa}{2}(U_cU_c^\dagger U_c-U_j U_c^\dagger U_cU_c^\dagger U_j), \quad t > 0, \\
& U_j \Big|_{t = 0} =U_j^{in} \in \mathbb{U}(d), \quad j = 1, \cdots, N.
\end{aligned}
\end{align}
\subsubsection{A gradient flow formulation} For an ensemble $\{ U_j \}_{j=1}^{N}$, consider the potential 
\begin{equation} \label{C-2}
\mathcal{V}_2(U) := -\frac{N\kappa}{4}  \mathrm{tr}((U_c U_c^\dagger)^2).
\end{equation}
\begin{lemma} \label{L3.2}
Let $\{U_j \}$ be a solution of \eqref{C-1} and ${\mathcal V}_2 = {\mathcal V}_2(U)$ be a potential defined by \eqref{C-2}. Then, one has
\[ \left.\frac{\partial\mathcal{V}_2}{\partial U_i}\right|_{T_{U_i}M_{d, d}(\mathbb{C})}=-\kappa U_cU_c^\dagger U_c. \]
\end{lemma}
\begin{proof} 
(i)~We set 
\[ [U_i]_{\alpha\beta} : =a_i^{\alpha\beta}+\mathrm{i}b_i^{\alpha\beta}, \quad a_i^{\alpha\beta},~b_i^{\alpha\beta} \in \bbr.   \]
Then we have
\begin{align}
\begin{aligned} \label{C-3}
\mathcal{V}_2&=-\frac{\kappa}{4N^3}\sum_{i, j, k, l}[U_i]_{\alpha\beta}[U_j^\dagger]_{\beta\gamma}[U_k]_{\gamma\delta}[U_l^\dagger]_{\delta\alpha}\\
&=-\frac{\kappa}{4N^3}\sum_{i, j, k, l}(a_i^{\alpha\beta}+\mathrm{i}b_i^{\alpha\beta})(a_j^{\gamma\beta}-\mathrm{i}b_j^{\gamma\beta})(a_k^{\gamma\delta}+\mathrm{i}b_k^{\gamma\delta})(a_l^{\alpha\delta}-\mathrm{i}b_l^{\alpha\delta}),
\end{aligned}
\end{align}
where we used Einstein summation rule. Now, we use the symmetry and \eqref{C-3} to get 
\begin{align}
\begin{aligned} \label{C-4}
\mathcal{V}_2&=-\frac{\kappa}{4N^3}\sum_{i, j, k, l}\left(a_i^{\alpha\beta}a_j^{\gamma\beta}a_k^{\gamma\delta}a_l^{\alpha\delta}+b_i^{\alpha\beta}b_j^{\gamma\beta}b_k^{\gamma\delta}b_l^{\alpha\delta}-2a_i^{\alpha\beta}b_j^{\gamma\beta}a_k^{\gamma\delta}b_l^{\alpha\delta} \right. \\
& \hspace{3cm}  \left. +~2a_i^{\alpha\beta}a_j^{\gamma\beta}b_k^{\gamma\delta}b_l^{\alpha\delta}+2b_i^{\alpha\beta}a_j^{\gamma\beta}a_k^{\gamma\delta}b_l^{\alpha\delta} \right ).
\end{aligned}
\end{align}
This yields
\begin{align*}
\frac{\partial {\mathcal V}_2}{\partial a_i^{\mu\nu}}&=-\frac{\kappa}{4N^3}\sum_{j, k, l}\left(4a_j^{\gamma\nu}a_k^{\gamma\delta}a_l^{\mu\delta}-4b_j^{\gamma\nu}a_k^{\gamma\delta}b_l^{\mu\delta}+4a_j^{\gamma\nu}b_k^{\gamma\delta}b_l^{\mu\delta}+4b_j^{\alpha\nu}a_k^{\mu\delta}b_l^{\alpha\delta}\right)\\
&=-\frac{\kappa}{N^3}\sum_{j, k, l}\left(a_j^{\mu\alpha}a_k^{\beta\alpha}a_l^{\beta\nu}-b_j^{\mu\alpha}a_k^{\beta\alpha}b_l^{\beta\nu}+b_j^{\mu\alpha}b_k^{\beta\alpha}a_l^{\beta\nu}+a_j^{\mu\alpha}b_k^{\beta\alpha}b_l^{\beta\nu}\right),\\
\frac{\partial {\mathcal V}_2}{\partial b_i^{\mu\nu}}&=-\frac{\kappa}{4N^3}\sum_{j, k, l}\left(4b_j^{\gamma\nu}b_k^{\gamma\delta}b_l^{\mu\delta}-4a_j^{\alpha\nu}a_k^{\mu\delta}b_l^{\alpha\delta}+4a_j^{\alpha\beta}a_k^{\mu\beta}b_l^{\alpha\nu}+4a_j^{\gamma\nu}a_k^{\gamma\delta}b_l^{\mu\delta}\right)\\
&=-\frac{\kappa}{N^3}\sum_{j, k, l}\left(b_j^{\mu\alpha}b_k^{\beta\alpha}b_l^{\beta\nu}-a_j^{\mu\alpha}b_k^{\beta\alpha}a_l^{\beta\nu}+a_j^{\mu\alpha}a_k^{\beta\alpha}b_l^{\beta\nu}+b_j^{\mu\alpha}a_k^{\beta\alpha}a_l^{\beta\nu}\right).
\end{align*}
Finally we can calculate
\begin{align*}
\left.\frac{\partial\mathcal{V}_2}{\partial U_i}\right|_{T_{U_i}M_{d, d}(\mathbb{C})}=\left(\frac{\partial {\mathcal V}_2}{\partial a_i^{\mu\nu}}+\mathrm{i}\frac{\partial {\mathcal V}_2}{\partial b_i^{\mu\nu}}\right)E^{\mu\nu},
\end{align*}
where $E^{\mu\nu}$ denotes the $d\times d$ matrix whose $(\mu,\nu)$-coordinate is 1 and the other coordinates are 0. By direct calculation, one has 
\begin{align*}
\begin{aligned}
\frac{\partial {\mathcal V}_2}{\partial a_i^{\mu\nu}}+\mathrm{i}\frac{\partial {\mathcal V}_2}{\partial b_i^{\mu\nu}} &=-\frac{\kappa}{N^3}\sum_{j, k, l}\big((a_j^{\mu\alpha}a_k^{\beta\alpha}a_l^{\beta\nu}-b_j^{\mu\alpha}a_k^{\beta\alpha}b_l^{\beta\nu}+b_j^{\mu\alpha}b_k^{\beta\alpha}a_l^{\beta\nu}+a_j^{\mu\alpha}b_k^{\beta\alpha}b_l^{\beta\nu})\\
& \hspace{1.5cm} +\mathrm{i}(b_j^{\mu\alpha}b_k^{\beta\alpha}b_l^{\beta\nu}-a_j^{\mu\alpha}b_k^{\beta\alpha}a_l^{\beta\nu}+a_j^{\mu\alpha}a_k^{\beta\alpha}b_l^{\beta\nu}+b_j^{\mu\alpha}a_k^{\beta\alpha}a_l^{\beta\nu})\big)\\
&=-\frac{\kappa}{N^3}\sum_{j, k, l}[U_j]_{\mu\alpha}[U_k^\dagger]_{\alpha\beta}[U_l]_{\beta\nu}.
\end{aligned}
\end{align*}
This implies
\[
\left.\frac{\partial\mathcal{V}_2}{\partial U_i}\right|_{T_{U_i}M_{d, d}(\mathbb{C})}=\left(\frac{\partial {\mathcal V}_2 }{\partial a_i^{\mu\nu}}+\mathrm{i}\frac{\partial {\mathcal V}_2}{\partial b_i^{\mu\nu}}\right)E^{\mu\nu}=-\frac{\kappa}{N^3}\sum_{j, k, l}U_jU_k^\dagger U_l=-\kappa U_cU_c^\dagger U_c.
\]
\end{proof}
\begin{proposition}  \label{P3.2}
 System $\eqref{C-1}_1$ can be rewritten as a gradient flow with the potential ${\mathcal V}_2$:
\[  \dot{U}_i = -\left. \frac{\partial\mathcal{V}_2}{\partial U_i}\right|_{T_{U_i}\mathbb{U}(d)}, \quad i = 1, \cdots, N. \]
\end{proposition}
\begin{proof}
We use Lemma \ref{L2.1} to see
\begin{align*}
\left.\frac{\partial\mathcal{V}_2}{\partial U_i}\right|_{T_{U_i}\mathbb{U}(d)}&=\pi_{U_i}\left(\left.\frac{\partial\mathcal{V}_2}{\partial U_i}\right|_{T_{U_i}M_{d, d}(\mathbb{C})}\right)=\pi\left(\left.\frac{\partial\mathcal{V}_2}{\partial U_i}\right|_{T_{U_i}M_{d, d}(\mathbb{C})}U_i^\dagger \right)U_i\\
&=\pi\left(-\frac{\kappa}{N^3}\sum_{j, k, l=1}^NU_jU_k^\dagger U_lU_i^\dagger \right)U_i=-\frac{\kappa}{2N^3}\sum_{j, k, l=1}^N(U_jU_k^\dagger U_lU_i^\dagger -U_iU_l^\dagger U_kU_j^\dagger)U_i\\
&=-\frac{\kappa}{2}(U_cU_c^\dagger U_c-U_iU_c^\dagger U_cU_c^\dagger U_i).
\end{align*}
\end{proof}
As a corollary of a gradient flow formulation of \eqref{C-1}, we have the convergence of the flow. 
\begin{corollary} \label{C3.1}
Suppose that coupling strength and the initial data $\{U_j^{in} \}$ satisfy 
\[  \kappa > 0, \quad  U_j^{in \dagger} U_j^{in} = I_d, \qquad  j = 1, \cdots, N, \]
and let $\{U_j \}$ be a global solution of system \eqref{C-1}. Then, there exists an equilibrium $(U_1^{\infty}, \cdots, U_N^{\infty})$ such that 
\[ \lim_{t \to \infty} \|U_j(t) - U_j^{\infty} \|_F = 0, \quad j = 1, \cdots, N. \]
\end{corollary}
\begin{proof}
Since system \eqref{C-1} is a gradient flow with the analytical potential ${\mathcal V}_2$, the flow $U_j$ converges toward an equilibrium (see Theorem 5.2 in \cite{H-K-R0}). 
\end{proof}

\subsubsection{Temporal evolution of potential} Next, we study temporal evolution of the potential ${\mathcal V}_1$ and ${\mathcal V}_2$ in the following lemma.
\begin{lemma} \label{L3.3}
Let $\{U_j \}$ be a global solution of system \eqref{C-1} with the initial data $\{U_j^{in} \}$:
\[ U_j^{in \dagger} U_j^{in} = I_d, \quad j = 1, \cdots, N. \] 
Then, one has
\begin{eqnarray*}
&& (i)~\frac{d}{dt} {\mathcal V}_1(U) = -\frac{\kappa^2}{4}\sum_{j=1}^N\|U_cU_c^\dagger U_j-U_cU_j^\dagger U_c\|^2_F-\frac{\kappa^2}{8}\sum_{j=1}^N\|U_cU_c^\dagger U_j-U_j U_c^\dagger U_c\|_F^2, \cr
&& (ii)~\frac{d}{dt} {\mathcal V}_2(U) =-\frac{\kappa^2}{4}\sum_{j=1}^N\|U_cU_c^\dagger U_cU_j^\dagger -U_j U_c^\dagger U_cU_c^\dagger \|_F^2.
\end{eqnarray*}
\end{lemma}
\begin{proof}
\noindent (i)~We use \eqref{C-1} to get 
\begin{align*}
\begin{aligned}
&\frac{d}{dt}\mathcal{V}_1(U)=-\frac{\kappa N}{2}\frac{d}{dt}\mathrm{tr}(U_cU_c^\dagger) =-\frac{\kappa N}{2}\mathrm{tr}(\dot{U}_cU_c^\dagger +U_c\dot{U}_c^\dagger)=-\frac{\kappa}{2}\sum_{j=1}^N\mathrm{tr}(\dot{U}_jU_c^\dagger+U_c\dot{U}_j^\dagger)  \\
&\hspace{0.5cm} =-\frac{\kappa^2}{4}\sum_{j=1}^N\left(\mathrm{tr}(U_cU_c^\dagger U_cU_c^\dagger -U_j U_c^\dagger U_cU_c^\dagger U_j U_c^\dagger)+(c.c.)\right) \\
&\hspace{0.5cm} = -\frac{\kappa^2}{4}\sum_{j=1}^N\mathrm{tr}(2U_cU_c^\dagger U_cU_c^\dagger -U_j U_c^\dagger U_cU_c^\dagger U_j U_c^\dagger -U_cU_j^\dagger U_cU_c^\dagger U_cU_j^\dagger)\\
&\hspace{0.5cm} =-\frac{\kappa^2}{4}\sum_{j=1}^N\mathrm{tr}\big((U_cU_c^\dagger U_j -U_cU_j^\dagger U_c)(U_j^\dagger U_cU_c^\dagger -U_c^\dagger U_j U_c^\dagger)+U_cU_c^\dagger U_cU_c^\dagger-U_cU_j^\dagger U_cU_c^\dagger U_j U_c^\dagger\big)\\
&\hspace{0.5cm} =-\frac{\kappa^2}{4}\sum_{j=1}^N\|U_cU_c^\dagger U_j-U_cU_j^\dagger U_c\|^2_F-\frac{\kappa^2}{8}\sum_{j=1}^N\|U_cU_c^\dagger U_j-U_j U_c^\dagger U_c\|_F^2.
\end{aligned}
\end{align*} 
(ii)~Similarly, one has
\begin{align*}
\begin{aligned}
&\frac{d}{dt}\mathcal{V}_2(U)=-\frac{\kappa N}{4}\frac{d}{dt}\mathrm{tr}(U_cU_c^\dagger U_cU_c^\dagger) =-\frac{\kappa N}{2}\left(\mathrm{tr}(\dot{U}_cU_c^\dagger U_cU_c^\dagger)+\mathrm{tr}(U_c\dot{U}_c^\dagger U_cU_c^\dagger)\right)\\
& \hspace{0.5cm}= -\frac{\kappa}{2}\sum_{j=1}^N \left(\mathrm{tr}(\dot{U}_jU_c^\dagger U_cU_c^\dagger)+\mathrm{tr}(U_c\dot{U}_j^\dagger U_cU_c^\dagger)\right)\\
& \hspace{0.5cm} =-\frac{\kappa^2}{4}\sum_{j=1}^N\left(\mathrm{tr}((U_cU_c^\dagger U_c-U_iU_c^\dagger U_cU_c^\dagger U_j)U_c^\dagger U_cU_c^\dagger)+(c.c)\right)\\
& \hspace{0.5cm}=-\frac{\kappa^2}{4}\sum_{j=1}^N\|U_cU_c^\dagger U_cU_j^\dagger -U_j U_c^\dagger U_cU_c^\dagger \|_F^2.
\end{aligned}
\end{align*}
\end{proof}
As a corollary, one has the following result. 

\begin{corollary} \label{C3.2}
Let $\{U_j \}$ be a global solution of system \eqref{C-1} with the initial data $\{U_j^{in} \}$:
\[ \kappa > 0, \quad (U_j^{in})^{\dagger} U_j^{in} = I_d, \quad j = 1, \cdots, N. \] 
Then, one has 
\[
\lim_{t \to \infty}  \|U_cU_c^\dagger U_cU_j^\dagger -U_j U_c^\dagger U_cU_c^\dagger \|_F^2 = 0, \quad j = 1, \cdots, N.  
\]
\end{corollary}
\begin{proof}
\noindent (i)~Since ${\mathcal V}_2(U)$ is bounded below and non-increasing along the flow \eqref{C-1}, ${\mathcal V}_2(U(\cdot))$ converges as $t \to \infty$. \newline

\noindent (ii)~It follows from Lemma \ref{L3.3} that 
\begin{equation} \label{C-4-0-0}
\frac{d}{dt} \mathcal{V}_2(U)  = -\frac{\kappa^2}{4}\sum_{j=1}^N\|U_cU_c^\dagger U_cU_j ^\dagger -U_j U_c^\dagger U_cU_c^\dagger\|_F^2. 
\end{equation}
In order to apply Babalat's lemma (Lemma \ref{L2.1}) for the derivation of the desired estimate, it suffices to show that 
\begin{equation} \label{C-4-0}
\sup_{0 \leq t < \infty} \Big|  \frac{d^2}{dt^2} \mathcal{V}_2(U)   \Big| < \infty. 
\end{equation}
We differentiate \eqref{C-4-0-0} with respect to $t$ and obtain
\begin{align*}
\frac{d^2}{dt^2} \mathcal{V}_2(U) &= -\frac{\kappa^2}{4}\sum_{j=1}^N \frac{d}{dt} \|U_cU_c^\dagger U_cU_j ^\dagger -U_j U_c^\dagger U_cU_c^\dagger\|_F^2.
\end{align*}
From the direct calculation, we have
\begin{align*}
&\frac{d}{dt} \|U_cU_c^\dagger U_cU_j ^\dagger -U_j U_c^\dagger U_cU_c^\dagger\|_F^2=\frac{d}{dt}\mathrm{tr}[(U_cU_c^\dagger U_cU_j ^\dagger -U_j U_c^\dagger U_cU_c^\dagger)(U_cU_c^\dagger U_cU_j ^\dagger -U_j U_c^\dagger U_cU_c^\dagger)^\dagger]\\
&=\frac{1}{N^3}\sum_{k_1, k_2, k_3, k_4, k_5, k_6=1}^N\frac{d}{dt}\mathrm{tr}[(U_{k_1}U_{k_2}^\dagger U_{k_3}U_j ^\dagger -U_j U_{k_1}^\dagger U_{k_2}U_{k_3}^\dagger)(U_{k_4}U_{k_5}^\dagger U_{k_6}U_j ^\dagger -U_j U_{k_4}^\dagger U_{k_5}U_{k_6}^\dagger)^\dagger].
\end{align*}
From the boundedness of $\|U_j\|_F,~\|U_{k_\alpha}\|_F,~\|U_j\|_{op},~\|U_{k_\alpha}\|_{op}$, $\|\dot{U}_j\|_F,~\|\dot{U}_{k_\alpha}\|_F,~\|\dot{U}_j\|_{op},~\|\dot{U}_{k_\alpha}\|_{op}$ and Lemma \ref{L2.2}, \ref{L2.3}, \ref{L2.4} we can obtain the boundedness of 
\[
\frac{d}{dt}\mathrm{tr}[(U_{k_1}U_{k_2}^\dagger U_{k_3}U_j ^\dagger -U_j U_{k_1}^\dagger U_{k_2}U_{k_3}^\dagger)(U_{k_4}U_{k_5}^\dagger U_{k_6}U_j ^\dagger -U_j U_{k_4}^\dagger U_{k_5}U_{k_6}^\dagger)^\dagger].
\]
Hence 
\[  \frac{d}{dt} \|U_cU_c^\dagger U_cU_j ^\dagger -U_j U_c^\dagger U_cU_c^\dagger\|_F^2 \quad \mbox{is uniformly bounded}. \]
Therefore, one has 
\[
\frac{d^2}{dt^2} \mathcal{V}_2(U) = -\frac{\kappa^2}{4}\sum_{j=1}^N \frac{d}{dt} \|U_cU_c^\dagger U_cU_j ^\dagger -U_j U_c^\dagger U_cU_c^\dagger\|_F^2.
\]
is uniformly bounded over time. So we can apply Barbalat's lemma to obtain
\[
\lim_{t\rightarrow\infty} \frac{d}{dt}\mathcal{V}_2(U)=0.
\]
This implies
\[
\lim_{t \to \infty}  \|U_cU_c^\dagger U_cU_j^\dagger -U_j U_c^\dagger U_cU_c^\dagger \|_F^2 = 0, \quad j = 1, \cdots, N.  
\]
\end{proof}

\subsection{Case with $m \geq 3$} \label{sec:3.2}
Consider the Cauchy problem for \eqref{C-0} in a mean-field form:
\begin{equation}
\begin{cases} \label{C-4-2}
\displaystyle \dot{U}_j = \frac{\kappa}{2} \Big(\underbrace{U_cU_c^\dagger U_c\cdots U_c^\dagger U_c}_{2m-1}-U_j\underbrace{U_c^\dagger U_cU_c^\dagger \cdots U_cU_c^\dagger}_{2m-1}U_j \Big),\\
\displaystyle U_j \Big|_{t = 0} =U_j^{in} \in\mathbb{U}(d).
\end{cases}
\end{equation}
Similar to Lemma \ref{L3.1} and Proposition \ref{P3.2}, one has a gradient flow formulation to \eqref{C-4-2}.
\begin{proposition}  \label{P3.3}
 System $\eqref{C-4-2}_1$ can be rewritten as a gradient flow with the potential ${\mathcal V}_m$:
\[  \dot{U}_j = -\left. \frac{\partial\mathcal{V}_m}{\partial U_j}\right|_{T_{U_j} \mathbb{U}(d)}, \quad j = 1, \cdots, N. \]
\end{proposition}
\begin{proof}
The proof is basically the same as in the proof of Proposition \ref{P3.2}. Hence we omit its details.
\end{proof}
As a corollary of a gradient flow formulation of \eqref{C-4-2}, we have the convergence of the flow. 
\begin{corollary} \label{C3.3}
Suppose that coupling strength and the initial data $\{U_j^{in} \}$ satisfy 
\[ \kappa > 0, \quad  U_j^{in \dagger} U_j^{in} = I_d, \qquad  j = 1, \cdots, N, \]
and let $\{U_j \}$ be a global solution of system \eqref{C-4-2}. Then, there exists an equilibrium $(U_1^{\infty}, \cdots, U_N^{\infty})$ such that 
\[ \lim_{t \to \infty} \|U_j(t) - U_j^{\infty} \|_F = 0, \quad j = 1, \cdots, N. \]
\end{corollary}
\begin{proof}
Since system \eqref{C-4-2} is a gradient flow with the analytical potential ${\mathcal V}_m(U)$, the flow $U_j(\cdot)$ converges toward an equilibrium (see Theorem 5.2 in \cite{H-K-R0}). 
\end{proof}

Now we want to find the derivative of functional $\mathcal{V}_m(U)$ along the dynamics \eqref{C-1}.
\begin{lemma} \label{L3.4}
Let $\{U_j \}$ be a global solution of system \eqref{C-4-2} with the initial data satisfying
\[ U_j^{in \dagger} U^{in}_j  = I_d, \quad j = 1, \cdots, N. \]
Then we have
\[
\frac{d}{dt}  {\mathcal V}_m(U)= -\frac{\kappa^2}{4}  \sum_{j=1}^N \|U_j \underbrace{U_c^\dagger U_cU_c^\dagger \cdots U_cU_c^\dagger}_{2m-1}-\underbrace{U_cU_c^\dagger U_c\cdots U_c^\dagger U_c}_{2m-1}U_j^\dagger \|_F^2.
\]
\end{lemma}
\begin{proof}
By direct calculations, one has
\begin{align*}
\begin{aligned}
& \frac{d}{dt}  {\mathcal V}_m(U) = -\frac{\kappa N}{2m}\frac{d}{dt}\mathrm{tr}((U_cU_c^\dagger)^m) =-\frac{\kappa N}{2}\left( \mathrm{tr}(\dot{U}_c\underbrace{U_c^\dagger \cdots U_cU_c^\dagger}_{2m-1})+(c.c.)\right)\\
& \hspace{0.5cm} = -\frac{\kappa^2}{4} \sum_{j=1}^N \mathrm{tr}((\underbrace{U_cU_c^\dagger U_c\cdots U_c^\dagger U_c}_{2m-1}-U_j \underbrace{U_c^\dagger U_cU_c^\dagger \cdots U_cU_c^\dagger}_{2m-1}U_i)\underbrace{U_c^\dagger \cdots U_cU_c^\dagger}_{2m-1})+(c.c.)\\
&\hspace{0.5cm} = -\frac{\kappa^2}{4} \sum_{j=1}^N \|U_j \underbrace{U_c^\dagger U_cU_c^\dagger \cdots U_cU_c^\dagger}_{2m-1}-\underbrace{U_cU_c^\dagger U_c\cdots U_c^\dagger U_c}_{2m-1}U_j^\dagger\|_F^2.
\end{aligned}
\end{align*}
\end{proof}
\begin{proposition} \label{P3.4}
Let $\{U_j \}$ be a global smooth solution of system \eqref{C-4-2} with the initial data $\{U_j^{in} \}$:
\[ U_j^{in \dagger} U_j^{in} = I_d, \quad j = 1, \cdots, N. \] 
Then, for $i = 1, \cdots, N,$
\[ \lim_{t \to \infty} \|U_i\underbrace{U_c^\dagger U_cU_c^\dagger \cdots U_cU_c^\dagger}_{2m-1}-\underbrace{U_cU_c^\dagger U_c\cdots U_c^\dagger U_c}_{2m-1}U_i^\dagger\|_F^2 = 0, \quad 
\lim_{t\rightarrow\infty} \| \dot{U}_j \|_F  =0. 
\]
\end{proposition}
\begin{proof}
\noindent (i)~The first assertion follows from the gradient flow formulation (Proposition \ref{P3.2}) and $U_j \in \bbu(d)$. \newline

\noindent (ii)~We use the boundedness of ${\mathcal V}_2$ (see \eqref{C-4-2-2}) and Lemma \ref{L3.4} to see
\begin{equation} \label{C-4-2-3}
 \exists~\lim_{t \to \infty} {\mathcal V}_m(U).
\end{equation} 
Note that
\[
\frac{d}{dt} {\mathcal V}_m(U) = -\frac{\kappa^2}{4}\sum_{i=1}^N \|U_i\underbrace{U_c^\dagger U_cU_c^\dagger \cdots U_cU_c^\dagger}_{2m-1}-\underbrace{U_cU_c^\dagger U_c\cdots U_c^\dagger U_c}_{2m-1}U_i^\dagger\|_F^2.
\]
We claim:
\begin{equation} \label{C-4-2-4}
 \sup_{0 \leq t < \infty} \Big| \frac{d^2}{dt^2} {\mathcal V}_m(U)  \Big| < \infty.
 \end{equation}
By \eqref{C-4-2-3} and \eqref{C-4-2-4}, we can apply Babalat's lemma to get the desired estimate:
\[  \lim_{t \to \infty} \frac{d}{dt} {\mathcal V}_m(U) = 0. \]
For the proof of claim \eqref{C-4-2-4}, it is sufficient to prove the uniform boundedness of
\[
\frac{d^2}{dt^2}\mathcal{V}_m(U).
\]
This proof is very similar to the proof of Corollary \ref{C3.2}, so we will omit. From this result, we have
\[
\lim_{t\rightarrow\infty}\frac{d}{dt} {\mathcal V}_m(U) = -\lim_{t\rightarrow\infty}\frac{\kappa^2}{4}\sum_{i=1}^N \|U_i\underbrace{U_c^\dagger U_cU_c^\dagger \cdots U_cU_c^\dagger}_{2m-1}-\underbrace{U_cU_c^\dagger U_c\cdots U_c^\dagger U_c}_{2m-1}U_i^\dagger\|_F^2=0.
\]
So we have
\[
\lim_{t\rightarrow\infty}\frac{\kappa^2}{4} \|U_i\underbrace{U_c^\dagger U_cU_c^\dagger \cdots U_cU_c^\dagger}_{2m-1}-\underbrace{U_cU_c^\dagger U_c\cdots U_c^\dagger U_c}_{2m-1}U_i^\dagger\|_F^2=0.
\]
for all $i=1, 2, \cdots, N$.\newline

\noindent (iii)~From the relation:
\[
\dot{U}_j = \frac{\kappa}{2}(\underbrace{U_cU_c^\dagger U_c\cdots U_c^\dagger U_c}_{2m-1}U_j^\dagger -U_j\underbrace{U_c^\dagger U_cU_c^\dagger \cdots U_cU_c^\dagger}_{2m-1})U_j,
\] 
we can transform above limit as follows:
\[
\|\dot{U}_j\|_F^2=\frac{\kappa^2}{4}\|U_j \underbrace{U_c^\dagger U_cU_c^\dagger \cdots U_cU_c^\dagger}_{2n-1}-\underbrace{U_cU_c^\dagger U_c\cdots U_c^\dagger U_c}_{2m-1}U_j^\dagger\|_F^2\rightarrow0\quad\mbox{as}\quad t\rightarrow \infty.
\]

\end{proof}

\begin{lemma} \label{L3.6}
Let $\{U_j\}$ be a global solution of the system \eqref{C-4} with $m=2^k$. Then we have 
\begin{align*}
\begin{aligned}
\frac{dR^2}{dt} &=\frac{\kappa}{2N}\sum_{i=1}^N \left(\|(U_cU_i^\dagger -U_iU_c^\dagger)\underbrace{U_cU_c^\dagger \cdots U_c}_{2^k-1}\|_F^2 \right. \\
&\hspace{2.5cm} \left. +\sum_{p=1}^{k}\frac{1}{2^p}\|\underbrace{U_cU_c^\dagger U_c\cdots U_c^\dagger}_{2^{k}}-\underbrace{U_cU_c^\dagger U_c\cdots U_c^\dagger}_{2^{k}-2^p}U_i\underbrace{U_c^\dagger \cdots U_c}_{2^p}U_i^\dagger\|_F^2\right).
\end{aligned}
\end{align*}
where $R^2 = \mbox{tr}(U_c U_c^{\dagger})$.
\end{lemma}
\begin{proof}
Note that
\begin{align*}
\frac{d}{dt}\mathrm{tr}(U_cU_c^\dagger)&=\mathrm{tr}(\dot{U}_cU_c^\dagger)+(c.c.)\\
&=\frac{\kappa}{2N}\sum_{i=1}^N\left(\mathrm{tr}((\underbrace{U_cU_c^\dagger U_c\cdots U_c^\dagger U_c}_{2^{k+1}-1}-U_i\underbrace{U_c^\dagger U_cU_c^\dagger \cdots U_cU_c^\dagger}_{2^{k+1}-1}U_i)U_c^\dagger)+(c.c.)\right).
\end{align*}
Here we have
\begin{align}
\begin{aligned} \label{C-4-2-5}
& \mathrm{tr}(\underbrace{U_cU_c^\dagger U_c \cdots U_c^\dagger U_cU_c^\dagger}_{2^{k+1}}-U_i \underbrace{U_c^\dagger U_cU_c^\dagger \cdots U_cU_c^\dagger}_{2^{k+1}-1}U_iU_c^\dagger)+(c.c.) \\
& \hspace{0.1cm} =\mathrm{tr}(\underbrace{U_cU_c^\dagger U_c\cdots U_c^\dagger}_{2^{k+1}-2}(2U_cU_c^\dagger -U_iU_c^\dagger U_iU_c^\dagger-U_cU_i^\dagger U_cU_i^\dagger))\\
&\hspace{0.1cm}  =\mathrm{tr}(\underbrace{U_cU_c^\dagger U_c\cdots U_c^\dagger}_{2^{k+1}-2}((U_cU_i^\dagger -U_iU_c^\dagger)(U_cU_i^\dagger-U_iU_c^\dagger)^\dagger+U_cU_c^\dagger-U_iU_c^\dagger U_cU_i^\dagger))\\
&\hspace{0.1cm}  =\|(U_cU_i^\dagger -U_iU_c^\dagger)\underbrace{U_cU_c^\dagger \cdots U_c}_{2^k-1}\|_F^2+\mathrm{tr}(\underbrace{U_cU_c^\dagger U_c\cdots U_c^\dagger}_{2^{k+1}}-\underbrace{U_cU_c^\dagger U_c\cdots U_c^\dagger}_{2^{k+1}-2}U_iU_c^\dagger U_cU_i^\dagger).
\end{aligned}
\end{align}
Now we define
\[
\mathcal{A}_{p} :=\mathrm{tr}(\underbrace{U_cU_c^\dagger U_c\cdots U_c^\dagger}_{2^{k+1}}-\underbrace{U_cU_c^\dagger U_c\cdots U_c^\dagger}_{2^{k+1}-2^p}U_i\underbrace{U_c^\dagger \cdots U_c}_{2^p}U_i^\dagger).
\]
Next, we derive a recursive relation between $\mathcal{A}_{p}$ and $\mathcal{A}_{p+1}$ when $1\leq p< k$:
\begin{align}
\begin{aligned}\label{C-5}
2\mathcal{A}_{p} &=2\mathrm{tr}(\underbrace{U_cU_c^\dagger U_c\cdots U_c^\dagger}_{2^{k+1}}-\underbrace{U_cU_c^\dagger U_c\cdots U_c^\dagger}_{2^{k+1}-2^p}U_i\underbrace{U_c^\dagger \cdots U_c}_{2^p}U_i^\dagger)\\
&=\|\underbrace{U_cU_c^\dagger U_c\cdots U_c^\dagger}_{2^{k}}-\underbrace{U_cU_c^\dagger U_c\cdots U_c^\dagger}_{2^{k}-2^p}U_i\underbrace{U_c^\dagger \cdots U_c}_{2^p}U_i^\dagger\|_F^2\\
&+
\mathrm{tr}(\underbrace{U_cU_c^\dagger U_c\cdots U_c^\dagger}_{2^{k+1}}-\underbrace{U_cU_c^\dagger U_c\cdots U_c^\dagger}_{2^{k}-2^p}U_i\underbrace{U_c^\dagger \cdots U_c}_{2^{p+1}}U_i^\dagger \underbrace{U_cU_c^\dagger U_c\cdots U_c^\dagger}_{2^{k}-2^p})\\
&=\|\underbrace{U_cU_c^\dagger U_c\cdots U_c^\dagger}_{2^{k}}-\underbrace{U_cU_c^\dagger U_c\cdots U_c^\dagger}_{2^{k}-2^p}U_i\underbrace{U_c^\dagger\cdots U_c}_{2^p}U_i^\dagger\|_F^2 \\
&+\mathrm{tr}(\underbrace{U_cU_c^\dagger U_c\cdots U_c^\dagger}_{2^{k+1}}-\underbrace{U_cU_c^\dagger U_c\cdots U_c^\dagger}_{2^{k+1}-2^{p+1}}U_i\underbrace{U_c^\dagger \cdots U_c}_{2^{p+1}}U_i^\dagger)\\
&=\|\underbrace{U_cU_c^\dagger U_c\cdots U_c^\dagger}_{2^{k}}-\underbrace{U_cU_c^\dagger U_c\cdots U_c^\dagger}_{2^{k}-2^p}U_i\underbrace{U_c^\dagger \cdots U_c}_{2^p}U_i^\dagger\|_F^2+\mathcal{A}_{p+1}.
\end{aligned}
\end{align}
On the other hand, $\mathcal{A}_k$ can be estimated as follows.
\begin{align}
\begin{aligned}\label{C-6}
\mathcal{A}_{k}&=\mathrm{tr}(\underbrace{U_cU_c^\dagger U_c\cdots U_c^\dagger}_{2^{k+1}}-\underbrace{U_c\cdots U_c^\dagger}_{2^{k}}U_i\underbrace{U_c^\dagger \cdots U_c}_{2^k}U_i^\dagger)\\
&=\frac{1}{2}\|\underbrace{U_cU_c^\dagger \cdots U_c^\dagger}_{2^{k}}U_i-U_i\underbrace{U_c^\dagger U_c\cdots U_c}_{2^{k}}\|_F^2 \\
&
=\frac{1}{2}\|\underbrace{U_cU_c^\dagger \cdots U_c^\dagger}_{2^{k}}-U_i\underbrace{U_c^\dagger U_c\cdots U_c}_{2^{k}}U_i^\dagger \|_F^2.
\end{aligned}
\end{align}
If we combine \eqref{C-5} and \eqref{C-6}, $\mathcal{A}_1$ can be calculated inductively.
\begin{align}
\begin{aligned} \label{C-6-1}
\mathcal{A}_1&=\frac{1}{2^1}\mathcal{A}_2+\frac{1}{2^1}\|\underbrace{U_cU_c^\dagger U_c\cdots U_c^\dagger}_{2^{k}}-\underbrace{U_cU_c^\dagger U_c\cdots U_c^\dagger}_{2^{k}-2^1}U_i\underbrace{U_c^\dagger \cdots U_c}_{2^1}U_i^\dagger\|_F^2\\
&=\cdots\\
&=\frac{1}{2^{k-1}}\mathcal{A}_k+\sum_{p=1}^{k-1}\frac{1}{2^p}\|\underbrace{U_cU_c^\dagger U_c\cdots U_c^\dagger}_{2^{k}}-\underbrace{U_cU_c^\dagger U_c\cdots U_c^\dagger}_{2^{k}-2^p}U_i\underbrace{U_c^\dagger \cdots U_c}_{2^p}U_i^\dagger \|_F^2\\
&=\sum_{p=1}^{k}\frac{1}{2^p}\|\underbrace{U_cU_c^\dagger U_c\cdots U_c^\dagger}_{2^{k}}-\underbrace{U_cU_c^\dagger U_c\cdots U_c^\dagger}_{2^{k}-2^p}U_i\underbrace{U_c^\dagger \cdots U_c}_{2^p}U_i^\dagger\|_F^2.
\end{aligned}
\end{align}
Finally, we combine \eqref{C-4-2-5} and \eqref{C-6-1} to get 
\begin{align*}
\begin{aligned}
&\|(U_cU_i^\dagger-U_iU_c^\dagger)\underbrace{U_cU_c^\dagger \cdots U_c}_{2^k-1}\|_F^2+\mathrm{tr}(\underbrace{U_cU_c^\dagger U_c\cdots U_c^\dagger}_{2^{k+1}}-\underbrace{U_cU_c^\dagger U_c\cdots U_c^\dagger}_{2^{k+1}-2}U_iU_c^\dagger U_cU_i^\dagger)\\
&\hspace{0.5cm} =\|(U_cU_i^\dagger-U_iU_c^\dagger)\underbrace{U_cU_c^\dagger \cdots U_c}_{2^k-1}\|_F^2+\mathcal{A}_2\\
&\hspace{0.5cm}  =\|(U_cU_i^\dagger-U_iU_c^\dagger)\underbrace{U_cU_c^\dagger \cdots U_c}_{2^k-1}\|_F^2 \\
&\hspace{0.5cm}+\sum_{p=1}^{k}\frac{1}{2^p}\|\underbrace{U_cU_c^\dagger U_c\cdots U_c^\dagger}_{2^{k}}-\underbrace{U_cU_c^\dagger U_c\cdots U_c^\dagger}_{2^{k}-2^p}U_i\underbrace{U_c^\dagger \cdots U_c}_{2^p}U_i^\dagger\|_F^2.
\end{aligned}
\end{align*}
From this, we have
\begin{align*}
\begin{aligned}
\frac{d}{dt}\mathrm{tr}(U_cU_c^\dagger) &=\frac{\kappa}{2N}\sum_{i=1}^N\left(\|(U_cU_i^\dagger-U_iU_c^\dagger)\underbrace{U_cU_c^\dagger\cdots U_c}_{2^k-1}\|_F^2  \right. \\
& \hspace{0.5cm} \left. + \sum_{p=1}^{k} \frac{1}{2^p} \|\underbrace{U_c U_c^\dagger U_c \cdots U_c^\dagger}_{2^{k}}-\underbrace{U_c U_c^\dagger U_c\cdots U_c^\dagger}_{2^{k}-2^p}U_i\underbrace{U_c^\dagger \cdots U_c}_{2^p}U_i^\dagger \|_F^2\right).
\end{aligned}
\end{align*}
\end{proof}
\begin{proposition} \label{P3.5}
Let $\{U_j\}$ be a global solution of system \eqref{C-4} with $m=2^k$. Then, one has 
\begin{align*}
&(i)~\lim_{t \to \infty} \|(U_cU_i^\dagger -U_iU_c^\dagger)\underbrace{U_cU_c^\dagger \cdots U_c}_{2^k-1}\|_F = 0. \\
&(ii)~\lim_{t \to \infty} \|\underbrace{U_cU_c^\dagger U_c\cdots U_c^\dagger}_{2^{k}}-\underbrace{U_cU_c^\dagger U_c\cdots U_c^\dagger}_{2^{k}-2^p}U_i\underbrace{U_c^\dagger \cdots U_c}_{2^p}U_i^*\|_F = 0,
\end{align*}
for all $p=1, 2, \cdots, k$, $i=1, 2, \cdots, N$.
\end{proposition}
\begin{proof}
It follows from Lemma \ref{L3.6} that $R$ is non-decreasing and bounded. Hence, $R$ tends to $R^{\infty}$ as $t \to \infty$. On the other hand, we use  the uniform boundedness of $\dot{U}_i$ and
\begin{align*}
\begin{aligned}
\frac{dR^2}{dt} &=\frac{\kappa}{2N}\sum_{i=1}^N\left(\|(U_cU_i^\dagger -U_iU_c^\dagger)\underbrace{U_cU_c^\dagger \cdots U_c}_{2^k-1}\|_F^2 \right. \\
& \hspace{1cm} \left.+\sum_{p=1}^{k}\frac{1}{2^p}\|\underbrace{U_cU_c^\dagger U_c\cdots U_c^\dagger}_{2^{k}}-\underbrace{U_cU_c^\dagger U_c\cdots U_c^\dagger}_{2^{k}-2^p}U_i\underbrace{U_c^\dagger \cdots U_c}_{2^p}U_i^\dagger \|_F^2\right)
\end{aligned}
\end{align*}
to show 
\[ \sup_{0 \leq t < \infty} \Big| \frac{d^2}{dt^2} R^2 \Big| < \infty. \]
Then, by Babalat's lemma, one has 
\[ \lim_{t \to \infty} \frac{dR^2}{dt}  = 0. \]
This implies the desired estimates.
\end{proof}

\section{A gradient flow formulation with a polynomial potential} \label{sec:4}
\setcounter{equation}{0}
In this section, we continue the study on the Lohe matrix model with higher-order couplings. In previous section, we considered the monomial potential function so that only one pair of coupling terms is involved in the coupling. In the sequel, we consider a polynomial potential function. \newline

Consider the Lohe matrix model in a mean-field form:
\begin{align}
\begin{aligned} \label{D-1}
&{\mathrm i} \dot{U}_j U_j^{\dagger} = \sum_{n=1}^m\frac{{\mathrm i} \kappa_n}{2}(\underbrace{U_cU_c^\dagger U_c\cdots U_c^\dagger U_c}_{2n-1} U_j^{\dagger}-U_j \underbrace{U_c^\dagger U_cU_c^\dagger \cdots U_cU_c^\dagger}_{2n-1}),\\
& U_i(0) =U_i^{in} \in\mathbb{U}(d).
\end{aligned}
\end{align}
First, we study a conservation law.
\begin{lemma} \label{L4.1}
Let $\{U_j\}$ be a global solution of  system \eqref{D-1}. Then  one has 
\[ \frac{d}{dt} (U^\dagger_j U_j) = 0, \quad t > 0,~~ j = 1, \cdots, N. \] 
\end{lemma}
\begin{proof} The proof is the same as that of Lemma \ref{L3.1}. Hence we omit its proof.
\end{proof}
For $U_j^{\dagger} U_j = U_j U_j^{\dagger}= I_d$, system \eqref{D-1} becomes 
\begin{equation} \label{D-1-1}
\dot{U}_j  = \sum_{n=1}^m\frac{\kappa_n}{2}\Big(\underbrace{U_cU_c^\dagger U_c\cdots U_c^\dagger U_c}_{2n-1} -U_j \underbrace{U_c^\dagger U_cU_c^\dagger \cdots U_cU_c^\dagger}_{2n-1}  U_j \Big).
\end{equation}
From now on, we assume
\[ U^{\dagger}_j U_j = U_j U_j^{\dagger} = I_d, \quad j = 1, \cdots, N. \]
Next, we study the gradient flow formulation of \eqref{D-1}. For this, we consider a polynomial potential:
\begin{equation*} \label{D-2}
\mathcal{V}_{poly}:=-N\mathrm{tr}(f(U_cU_c^\dagger)), \quad f(A) :=\frac{\kappa_1}{2}A+\frac{\kappa_2}{4}A^2+\cdots+\frac{\kappa_m}{2m}A^m.
\end{equation*}
Then, $\mathcal{V}_{poly}$ is an analytic function and  since 
\[
\left| \mbox{tr} \left[( U_cU_c^\dagger)^n\right] \right| \leq  {\underbrace{\left\|U_cU_c^\dagger \cdots  \right\|_F^2}_{n-\mbox{times of } U_c}}\leq \left( \|U_c\|_{op}^{n-1}\cdot \|U_c\|_F\right)^2\leq d.
\]
it is easy to see 
\begin{equation*} \label{D-1-2}
|\mathcal{V}_{poly}| \leq \frac{N}{2}\left(\kappa_1+\kappa_2+\cdots+\kappa_m\right)d.
\end{equation*}

\begin{proposition}  \label{P4.1}
 System $\eqref{D-1-1}$ can be rewritten as a gradient flow with potential ${\mathcal V}_{poly}$:
\[  \dot{U}_j = -\left. \frac{\partial\mathcal{V}_{poly}}{\partial U_j}\right|_{T_{U_j} \mathbb{U}(d)}, \quad j = 1, \cdots, N. \]
\end{proposition}
\begin{proof}
The proof is basically the same as in the proof of Proposition \ref{P3.2}. Hence we omit its details.
\end{proof}

\begin{lemma} \label{L4.2}
Let $\{U_j \}$ be a global solution of system \eqref{D-1} with the initial data satisfying
\[ U_j^{in \dagger} U^{in}_j  = I_d, \quad j = 1, \cdots, N. \]
Then, one has 
\[
\frac{d}{dt} \mathcal{V}_{poly}= -\sum_{i=1}^N\left\| \sum_{n=1}^m \frac{\kappa_n}{2} \left((U_cU_c^\dagger)^{n-1}U_cU_i^\dagger-U_iU_c^\dagger(U_cU_c^\dagger)^{n-1}\right)\right\|_F^2.
\]
\end{lemma}
\begin{proof}  We use \eqref{D-1-1} to see
\begin{equation} \label{D-3}
\frac{d}{dt}\mathrm{tr}(f(U_cU_c^\dagger)) =\sum_{n=1}^m\frac{\kappa_n}{2n}\frac{d}{dt}\mathrm{tr}((U_cU_c^\dagger)^n) =\sum_{n=1}^m\frac{\kappa_n}{2}\left(\mathrm{tr}(\dot{U}_cU_c^\dagger(U_cU_c^\dagger)^{n-1})+(c.c.)\right).
\end{equation}
The first term in the R.H.S. of \eqref{D-5} can be estimated as follows.
\begin{align}
\begin{aligned} \label{D-4}
&\mathrm{tr}(\dot{U}_cU_c^\dagger(U_cU_c^\dagger)^{n-1}) \\
& \hspace{0.2cm} =\frac{1}{N}\sum_{i=1}^N
\sum_{l=1}^m\frac{\kappa_l}{2}\mathrm{tr}\big((\underbrace{U_cU_c^\dagger U_c\cdots U_c^\dagger U_c}_{2l-1}-U_i\underbrace{U_c^\dagger U_cU_c^\dagger \cdots U_cU_c^\dagger}_{2l-1}U_i)U_c^\dagger (U_cU_c^\dagger)^{n-1}\big)\\
&\hspace{0.2cm} =\frac{1}{N}\sum_{i=1}^N\sum_{l=1}^m\frac{\kappa_l}{2}\mathrm{tr}\big((U_cU_c^\dagger)^{n+l-1}-U_iU_c^\dagger (U_cU_c^\dagger)^{l-1}U_iU_c^\dagger (U_cU_c^\dagger)^{n-1}\big).
\end{aligned}
\end{align}
We combine \eqref{D-3} and \eqref{D-4} to obtain

\begin{align}
\begin{aligned} \label{D-5}
&\frac{d}{dt}\mathrm{tr}(f(U_cU_c^\dagger))\\
&\hspace{0.2cm}=\sum_{n=1}^m\frac{\kappa_n}{2}\left(\mathrm{tr}(\dot{U}_cU_c^\dagger(U_cU_c^\dagger)^{n-1})+(c.c.)\right)\\
&\hspace{0.2cm}=\frac{1}{N}\sum_{i=1}^N\sum_{l, n=1}^m\frac{\kappa_l\kappa_n}{4}\mathrm{tr}\big((U_cU_c^\dagger)^{n+l-1}-U_iU_c^\dagger (U_cU_c^\dagger)^{l-1}U_iU_c^\dagger (U_cU_c^\dagger)^{n-1}\big)+(c.c.)\\
&\hspace{0.2cm}=\frac{1}{N}\sum_{i=1}^N\sum_{l, n=1}^m\frac{\kappa_l\kappa_n}{4}\mathrm{tr}\left[
\left((U_cU_c^\dagger)^{n-1}U_cU_i^\dagger-U_iU_c^\dagger(U_cU_c^\dagger)^{n-1}\right)\left((U_cU_c^\dagger)^{l-1}U_cU_i^\dagger-U_iU_c^\dagger(U_cU_c^\dagger)^{l-1}\right)^\dagger\right]\\
&\hspace{0.2cm}=\frac{1}{N}\sum_{i=1}^N\left\| \sum_{n=1}^m \frac{\kappa_n}{2} \left((U_cU_c^\dagger)^{n-1}U_cU_i^\dagger-U_iU_c^\dagger(U_cU_c^\dagger)^{n-1}\right)\right\|_F^2.
\end{aligned}
\end{align}
Therefore, we have following equality:
\[
\frac{d}{dt} \mathcal{V}_{poly}=-N\frac{d}{dt}\mathrm{tr}(f(U_cU_c^\dagger))=-\sum_{i=1}^N\left\| \sum_{n=1}^m \frac{\kappa_n}{2} \left((U_cU_c^\dagger)^{n-1}U_cU_i^\dagger-U_iU_c^\dagger(U_cU_c^\dagger)^{n-1}\right)\right\|_F^2.
\]
\end{proof}
\begin{remark}
Since
\[
f'(A)=\frac{1}{2}\left(\kappa_1I+\kappa_2 A+\cdots+\kappa_m A^{m-1}\right),
\]
we can express above result as follows:
\begin{align*}
\sum_{i=1}^N\left\| \sum_{n=1}^m \frac{\kappa_n}{2} \left((U_cU_c^\dagger)^{n-1}U_cU_i^\dagger-U_iU_c^\dagger(U_cU_c^\dagger)^{n-1}\right)\right\|_F^2=\sum_{i=1}^N\left\| f'(U_cU_c^\dagger)U_cU_i^\dagger-U_iU_c^\dagger f'(U_cU_c^\dagger)\right\|_F^2.
\end{align*}
\end{remark}

\begin{theorem} \label{T4.1}
Let $\{U_j \}$ be a global solution of system \eqref{D-1} with the initial data $\{U_j^{in} \}$:
\[ U_j^{in \dagger} U_j^{in} = I_d, \quad j = 1, \cdots, N. \] 
Then, there exists an equilibrium $(U_1^{\infty}, \cdots, U_N^{\infty})$ such that 
\[ \lim_{t \to \infty} \|U_j(t) - U_j^{\infty} \|_F = 0, \qquad \lim_{t \to \infty} \frac{d}{dt} {\mathcal V}_{poly}(U) = 0, \qquad  \lim_{t\rightarrow\infty} \| \dot{U}_j \|_F 
=0, \quad   j = 1, \cdots, N. \]
\end{theorem}
\begin{proof}
\noindent (i) By Lemma \ref{L4.1} and the assumption on the initial data, we have
\[ U_j^{\dagger} U_j = U_j U_j^{\dagger}= I_d, \quad j= 1, \cdots, N. \]
Under this circumstance, dynamics of \eqref{D-1} is equivalent to \eqref{D-1-1}. Moreover, it follows from Proposition \ref{P4.1} and analyticity of the potential function that the ensemble $(U_1, \cdots, U_N)$ tends to an equilibrium $(U_1^{\infty}, \cdots, U_N^{\infty})$ as $t \to \infty$. \newline

\noindent (ii)~We can use similar argument to prove the boundedness of
\[
\sup_{0 \leq t < \infty} \Big|  \frac{d^2}{dt^2} {\mathcal V}_{poly} \Big|.
\]
Then we can apply the Barbalat's lemma to obtain
\[
\lim_{t\rightarrow\infty} \frac{d}{dt} {\mathcal V}_{poly} =0.
\]
(iii)~Above result yields,
\[
\lim_{t\rightarrow\infty}\frac{d}{dt} \mathcal{V}_{poly}=-\lim_{t\rightarrow\infty}\sum_{i=1}^N\left\| \sum_{n=1}^m \frac{\kappa_n}{2} \left((U_cU_c^\dagger)^{n-1}U_cU_i^\dagger-U_iU_c^\dagger(U_cU_c^\dagger)^{n-1}\right)\right\|_F^2=0.
\]
This also implies
\[
\left\| \sum_{n=1}^m \frac{\kappa_n}{2} \left((U_cU_c^\dagger)^{n-1}U_cU_i^\dagger-U_iU_c^\dagger(U_cU_c^\dagger)^{n-1}\right)\right\|_F^2\rightarrow0\quad\mbox{as}\quad t\rightarrow\infty.
\]
If we combine the above relation and following relation
\[
\dot{U}_i= \sum_{n=1}^m\frac{\kappa_n}{2}\Big(\underbrace{U_cU_c^\dagger U_c\cdots U_c^\dagger U_c}_{2n-1} -U_i \underbrace{U_c^\dagger U_cU_c^\dagger \cdots U_cU_c^\dagger}_{2n-1}  U_i \Big),
\]
we have
\begin{align*}
\|\dot{U}_i\|^2_F=\left\| \sum_{n=1}^m \frac{\kappa_n}{2} \left((U_cU_c^\dagger)^{n-1}U_cU_i^\dagger-U_iU_c^\dagger(U_cU_c^\dagger)^{n-1}\right)\right\|_F^2\rightarrow0\quad\mbox{as}\quad t\rightarrow\infty.
\end{align*}
\end{proof}
Next, we consider the following special polynomial type function $f$ satisfying the following property: 
\[
\kappa_j\neq0\Leftarrow j=2^n \quad \mbox{for some $n\in\mathbb{N}$}. 
\]
i.e., $f(A)$ takes the following form:
\[
f(A)=\frac{\kappa_{2^0}}{2^1}A^{2^0}+\frac{\kappa_{2^1}}{2^2}A^{2^1}+\cdots+\frac{\kappa_{2^{l-1}}}{2^{l}}A^{2^{l-1}}.
\]
Then we have following system:
\begin{align}\label{D-7}
\begin{cases}
\dot{U}_j =\displaystyle\sum_{k=0}^{l-1}\frac{\kappa_{2^k}}{2}(\underbrace{U_cU_c^\dagger U_c\cdots U_c^\dagger U_c}_{2^{k+1}-1}-U_j \underbrace{U_c^\dagger U_cU_c^\dagger \cdots U_cU_c^\dagger}_{2^{k+1}-1}U_j), \quad t >0,\\
U_j(0) =U_j^{in}\in\mathbb{U}(d), \quad j = 1, \cdots, N.
\end{cases}
\end{align}
We have following dynamics of order parameter.
\begin{lemma} \label{L4.3}
Let $\{U_i\}$ be a global solution of system \eqref{D-7}. Then we have
\begin{align*}
&\frac{dR^2}{dt} =\sum_{k=0}^{l-1}\sum_{i=1}^N\frac{\kappa_{2^k}}{2N}\left(\|(U_cU_i^\dagger-U_iU_c^\dagger)\underbrace{U_cU_c^\dagger \cdots U_c}_{2^k-1}\|_F^2 \right. \\
& \hspace{3.5cm} + \left. \sum_{p=1}^{k}\frac{1}{2^p}\|\underbrace{U_cU_c^\dagger U_c\cdots U_c^\dagger}_{2^{k}}-\underbrace{U_cU_c^\dagger U_c\cdots U_c^\dagger}_{2^{k}-2^p}U_i\underbrace{U_c^\dagger \cdots U_c}_{2^p}U_i^\dagger\|_F^2\right).
\end{align*}
\end{lemma}
\begin{proof} By direct calculations, one has 
\begin{align*}
\begin{aligned}
\frac{d}{dt}\|U_c\|_F^2 &=\mathrm{tr}(\dot{U}_iU_i^\dagger)+(c.c.)\\
&=\sum_{k=0}^{l-1}\frac{\kappa_{2^k}}{2}\mathrm{tr}\big((\underbrace{U_cU_c^\dagger U_c\cdots U_c^\dagger U_c}_{2^{k+1}-1}-U_i\underbrace{U_c^\dagger U_cU_c^\dagger \cdots U_cU_c^\dagger}_{2^{k+1}-1}U_i)U_c^\dagger \big)\\
&=\sum_{k=0}^{l-1}\sum_{i=1}^N\frac{\kappa_{2^k}}{2N}\left(\|(U_cU_i^\dagger-U_iU_c^\dagger)\underbrace{U_cU_c^\dagger \cdots U_c}_{2^k-1}\|_F^2 \right. \\
& \hspace{2cm} \left. +\sum_{p=1}^{k}\frac{1}{2^p}\|\underbrace{U_cU_c^\dagger U_c\cdots U_c^\dagger}_{2^{k}}-\underbrace{U_cU_c^\dagger U_c\cdots U_c^\dagger}_{2^{k}-2^p}U_i\underbrace{U_c^\dagger \cdots U_c}_{2^p}U_i^\dagger\|_F^2\right).
\end{aligned}
\end{align*}
\end{proof}
\begin{theorem} \label{T4.2}
Let $\{U_j \}$ be a global solution of system \eqref{D-7} with the initial data $\{U_j^{in} \}$:
\[ U_j^{in \dagger} U_j^{in} = I_d, \quad j = 1, \cdots, N. \] 
Then, the following assertions hold.
\begin{enumerate}
\item
For all $i=1, 2, \cdots, N$ and for all $k=0, 1, \cdots, l-1$ which satisfies $\kappa_{2^k}\neq0$, 
\[
\lim_{t\rightarrow\infty}\|(U_cU_i^\dagger -U_iU_c^\dagger)\underbrace{U_cU_c^\dagger \cdots U_c}_{2^k-1}\|_F=0.
\]
\item
For all $i=1, 2, \cdots, N$, for all $p=1, 2, \cdots, k$ and for all $k=0, 1, \cdots, l-1$ which satisfies $\kappa_{2^k}\neq0$, 
\[
\lim_{t\rightarrow\infty}\|\underbrace{U_cU_c^\dagger U_c\cdots U_c^\dagger}_{2^{k}}-\underbrace{U_cU_c^\dagger U_c\cdots U_c^\dagger}_{2^{k}-2^p}U_i\underbrace{U_c^\dagger \cdots U_c}_{2^p}U_i^\dagger\|_F=0.
\]
\end{enumerate}
\end{theorem}
\begin{proof}
Since $R^2$ is bounded and non-increasing, $R^2$ converges as $t \to \infty$. Next, we will show
\[
\lim_{t \to \infty} \frac{dR^2}{dt} = 0.
\]
For this, it suffices to check 
\[ \sup_{0 \leq t < \infty} \Big|  \frac{d^2 R^2}{dt^2} \Big| < \infty. \] 
Once the above estimate is verified, then Babalat's lemma yields the desired estimates. However the proof of the boundedness of second derivative of $R^2$ is very similar to the proof of Corollary \eqref{C3.2}. Then we have
\begin{align*}
\lim_{t \to \infty} \frac{dR^2}{dt}&=\lim_{t \to \infty} \sum_{k=0}^{l-1}\sum_{i=1}^N\frac{\kappa_{2^k}}{2N}\left(\|(U_cU_i^\dagger-U_iU_c^\dagger)\underbrace{U_cU_c^\dagger \cdots U_c}_{2^k-1}\|_F^2 \right. \\
& \hspace{2cm} \left. +\sum_{p=1}^{k}\frac{1}{2^p}\|\underbrace{U_cU_c^\dagger U_c\cdots U_c^\dagger}_{2^{k}}-\underbrace{U_cU_c^\dagger U_c\cdots U_c^\dagger}_{2^{k}-2^p}U_i\underbrace{U_c^\dagger \cdots U_c}_{2^p}U_i^\dagger\|_F^2\right) = 0.
\end{align*}
From this equality, for all $i$ and $k$ which satisfies $\kappa_{2^k}$, we have
\begin{align*}
\begin{aligned}
& \lim_{t\rightarrow\infty}\left(\|(U_cU_i^\dagger-U_iU_c^\dagger)\underbrace{U_cU_c^\dagger \cdots U_c}_{2^k-1}\|_F^2  +\sum_{p=1}^{k}\frac{1}{2^p}\|\underbrace{U_cU_c^\dagger U_c\cdots U_c^\dagger}_{2^{k}}-\underbrace{U_cU_c^\dagger U_c\cdots U_c^\dagger}_{2^{k}-2^p}U_i\underbrace{U_c^\dagger \cdots U_c}_{2^p}U_i^\dagger\|_F^2\right) \\
&= 0.
\end{aligned}
\end{align*}
Since each term is non-negative, each term must converge to zero, we have
\[
\|(U_cU_i^\dagger-U_iU_c^\dagger)\underbrace{U_cU_c^\dagger \cdots U_c}_{2^k-1}\|_F\rightarrow0\quad\mbox{as}\quad t\rightarrow\infty
\]
and
\[
\|\underbrace{U_cU_c^\dagger U_c\cdots U_c^\dagger}_{2^{k}}-\underbrace{U_cU_c^\dagger U_c\cdots U_c^\dagger}_{2^{k}-2^p}U_i\underbrace{U_c^\dagger \cdots U_c}_{2^p}U_i^\dagger\|_F\quad\mbox{as}\quad t\rightarrow\infty
\]
for all $i=1, 2, \cdots, N$, $p=1, 2, \cdots, k$, and $k=0, 1, \cdots, l-1$ which satisfies $\kappa_{2^k}\neq0$.
\end{proof}
\section{Emergent dynamics of Lohe ensemble} \label{sec:5}
\setcounter{equation}{0}
In this section, we study a relaxation estimate toward the aggregated state for system \eqref{D-1}. In previous section, we show that the state configuration tends to an equilibrium for any initial data without any explicit decay estimate. The main reason for this is that we employed a gradient flow approach and Babalat's lemma which does not tell us any constructive decay estimate. For an explicit decay estimate, we employ a diameter functional and derive a Riccati type differential inequality for the state diameter. This yields an explicit decay estimate for some restricted class of initial data and system parameters. 

\subsection{Ensemble diameter} \label{sec:5.1}
For a state configuration $\{U_j \}$, we set 
\[ {\mathcal D}(U) := \max_{i,j} \| U_i-U_j \|_F.
\]
\begin{lemma} \label{L5.1}
Let $\{U_j \}$ be a global solution to system \eqref{D-1}. Then ${\mathcal D}(U)$ satisfies 
\begin{align*}\label{E-1}
-\kappa_+D(U)^2-\kappa_1D(U)^4\leq\frac{d}{dt}D(U)^2\leq-\kappa_-D(U)^2+\kappa_1D(U)^4.
\end{align*}
where $\kappa_+$ and $\kappa_-$ are given by the following relations:
\[
\kappa_-=2\kappa_1-\sqrt{d}\sum_{n=2}^m\kappa_n \quad \mbox{and} \quad \kappa_+=2\kappa_1+\sqrt{d}\sum_{n=2}^m\kappa_n.
\]
\end{lemma}
\begin{proof}
Let $(i, j)$ be a pair of indices. By direct estimate, one has
\begin{align*}
\begin{aligned} 
&\frac{d}{dt}\|U_i-U_j\|_F^2  \\
&  =\frac{d}{dt}\mathrm{tr}(2I-U_iU_j^\dagger -U_jU_i^\dagger ) =-\mathrm{tr}(\dot{U}_iU_j^\dagger +\dot{U}_j U_i^\dagger)-(c.c.)\\
&  =-\sum_{n=1}^m\frac{\kappa_n}{2}\mathrm{tr}\big((\underbrace{U_cU_c^\dagger U_c\cdots U_c^\dagger  U_c}_{2n-1}-U_i\underbrace{U_c^\dagger U_cU_c^\dagger  \cdots U_cU_c^\dagger}_{2^n-1}U_i)U_j^\dagger  \big)\\
& -\sum_{n=1}^m\frac{\kappa_n}{2}\mathrm{tr}\big((\underbrace{U_cU_c^\dagger U_c\cdots U_c^\dagger U_c}_{2n-1}-U_j\underbrace{U_c^\dagger U_cU_c^\dagger \cdots U_cU_c^\dagger}_{2n-1}U_j)U_i^\dagger \big)-(c.c.)\\
& =-\sum_{n=1}^m\frac{\kappa_n}{2}\mathrm{tr}\big((\underbrace{U_cU_c^\dagger U_c\cdots U_c^\dagger U_c}_{2n-1})(U_i^\dagger +U_j^\dagger ) -(\underbrace{U_c^\dagger U_cU_c^\dagger \cdots U_cU_c^\dagger }_{2n-1})(U_iU_j^\dagger U_i-U_jU_i^\dagger U_j)\big)-(c.c.)\\
&=-\sum_{n=1}^m\frac{\kappa_n}{2}\mathrm{tr}\big((\underbrace{U_cU_c^\dagger U_c\cdots U_c^\dagger  U_c}_{2n-1})(U_i^\dagger +U_j^\dagger -U_i^\dagger U_jU_i^\dagger -U_j^\dagger U_iU_j^\dagger )\big)-(c.c.)\\
& =-\sum_{n=1}^m\frac{\kappa_n}{2}\mathrm{tr}\big((U_cU_c^\dagger )^{n-1}(U_cU_i^\dagger +U_cU_j^\dagger -U_cU_i^\dagger  U_jU_i^\dagger \\
&-U_cU_j^\dagger  U_iU_j^\dagger +U_iU_c^\dagger +U_jU_c^\dagger -U_iU_j^\dagger  U_iU_c^\dagger -U_jU_i^\dagger U_jU_c^\dagger )\big).
\end{aligned}
\end{align*}
It follows from the Lemma \ref{L2.2} that 
\begin{align}
\begin{aligned} \label{E-3}
&\Big|\mathrm{tr} \Big((U_cU_c^\dagger)^{n-1}(U_cU_i^\dagger+U_cU_j^\dagger-U_cU_i^\dagger U_jU_i^\dagger-U_cU_j^\dagger U_iU_j^\dagger+U_iU_c^\dagger \\
& \hspace{5cm} +U_jU_c^\dagger-U_iU_j^\dagger U_iU_c^\dagger-U_jU_i^\dagger U_jU_c^\dagger)\Big) \Big|\\
& \hspace{1cm} \leq \sqrt{d} \Big \|(U_cU_c^\dagger)^{n-1}(U_cU_i^\dagger+U_cU_j^\dagger-U_cU_i^\dagger U_jU_i^\dagger-U_cU_j^\dagger U_iU_j^\dagger \\
& \hspace{5cm} +U_iU_c^\dagger+U_jU_c^\dagger-U_iU_j^\dagger U_iU_c^\dagger-U_jU_i^\dagger U_jU_c^\dagger)\|_F.
\end{aligned}
\end{align}
On the other hand, Lemma \ref{L2.3} and Lemma \ref{L2.4} imply
\begin{align}
\begin{aligned} \label{E-4}
&\sqrt{d}\|(U_cU_c^\dagger)^{n-1}(U_cU_i^\dagger+U_cU_j^\dagger-U_cU_i^\dagger U_jU_i^\dagger-U_cU_j^\dagger U_iU_j^\dagger \\
& \hspace{5cm} +U_iU_c^\dagger +U_jU_c^\dagger -U_iU_j^\dagger U_iU_c^\dagger-U_jU_i^\dagger U_jU_c^\dagger)\|_F\\
& \hspace{1cm} \leq\sqrt{d}\|U_c\|_{op}^{2n-2}\|U_cU_i^\dagger +U_cU_j^\dagger -U_cU_i^\dagger U_jU_i^\dagger -U_cU_j^\dagger U_iU_j^\dagger \\
& \hspace{5cm} +U_iU_c^\dagger +U_jU_c^\dagger -U_iU_j^\dagger U_iU_c^\dagger-U_jU_i^\dagger U_jU_c^\dagger \|_F\\
& \hspace{1cm} \leq\sqrt{d}\|U_cU_i^\dagger +U_cU_j^\dagger-U_cU_i^\dagger U_jU_i^\dagger -U_cU_j^\dagger U_iU_j^\dagger +U_iU_c^\dagger \\
& \hspace{5cm} +U_jU_c^\dagger -U_iU_j^\dagger U_iU_c^\dagger -U_jU_i^\dagger U_jU_c^\dagger\|_F.
\end{aligned}
\end{align}
Note that 
\begin{align}
\begin{aligned} \label{E-5}
& U_cU_i^\dagger+U_cU_j^\dagger -U_cU_i^\dagger U_jU_i^\dagger -U_cU_j^\dagger U_iU_j^\dagger \\
& \hspace{2cm} =U_c(U_i-U_j)^\dagger(U_i-U_j)U_i^\dagger -U_cU_j^\dagger (U_i-U_j)(U_j-U_i)^\dagger.
\end{aligned}
\end{align}
Then, one has
\begin{align}
\begin{aligned} \label{E-6}
&\|U_cU_i^\dagger +U_cU_j^\dagger -U_cU_i^\dagger U_jU_i^\dagger -U_cU_j^\dagger U_iU_j^\dagger +U_iU_c^\dagger +U_jU_c^\dagger -U_iU_j^\dagger U_iU_c^\dagger -U_jU_i^\dagger U_jU_c^\dagger \|_F\\
&\hspace{0.5cm} \leq \|U_c(U_i-U_j)^\dagger (U_i-U_j)U_i^\dagger \|_F+\|U_cU_j^\dagger (U_i-U_j)(U_j-U_i)^\dagger\|_F\\
&\hspace{0.5cm} \leq 2\|U_c\|_{op}\cdot\|(U_i-U_j)^*(U_i-U_j)\|_F\leq2\|(U_i-U_j)^\dagger (U_i-U_j)\|_F.
\end{aligned}
\end{align}
From the Lemma \ref{L2.4}, we have
\[
\|(U_i-U_j)^\dagger (U_i-U_j)\|_F\leq\|U_i-U_j\|_F^2.
\]
Finally, we combine \eqref{E-3}, \eqref{E-4}, \eqref{E-5} and \eqref{E-6} to get 
\begin{align*}
\begin{aligned}
&\Big |\mathrm{tr}\big((U_cU_c^\dagger)^{n-1}(U_cU_i^\dagger +U_cU_j^\dagger -U_cU_i^\dagger U_jU_i^\dagger -U_cU_j^\dagger U_iU_j^\dagger \\
& \hspace{3cm} +U_iU_c^\dagger +U_jU_c^\dagger -U_iU_j^\dagger U_iU_c^\dagger -U_jU_i^\dagger U_jU_c^\dagger)\big) \Big| \leq 2\sqrt{d}\|U_i-U_j\|_F^2.
\end{aligned}
\end{align*}
It follows from \cite{H-R} that we have following estimate with $n=1$. If we set
\begin{align*}
\begin{aligned}
\mathcal{I}_n &:=\frac{\kappa_{n}}{2}\mathrm{tr}\Big((U_cU_c^\dagger)^{n-1}(U_cU_i^\dagger +U_cU_j^\dagger -U_cU_i^\dagger U_jU_i^\dagger -U_cU_j^\dagger U_iU_j^\dagger \\
&\hspace{2cm} +U_iU_c^\dagger +U_jU_c^\dagger-U_iU_j^\dagger U_iU_c^\dagger -U_jU_i^\dagger U_jU_c^\dagger)\Big),
\end{aligned}
\end{align*}
then we have
\[
-2\kappa_1D(U)^2-\kappa_1D(U)^4\leq-\mathcal{I}_1\leq -2\kappa_1D(U)^2+\kappa_1D(U)^4.
\]
For $n>1$, we have
\[
|\mathcal{I}_n|\leq \kappa_{n}\sqrt{d}D(U)^2.
\]
From the equality:
\[
-\sum_{n=0}^{m}\mathcal{I}_n=\frac{d}{dt}D(U)^2,
\]
we have following estimate:
\begin{align*}
\begin{aligned}
& -\left(2\kappa_1+\sqrt{d}\sum_{n=2}^m\kappa_n\right)D(U)^2-\kappa_1D(U)^4\leq\frac{d}{dt}D(U)^2 \\
& \hspace{3cm} \leq-\left(2\kappa_1-\sqrt{d}\sum_{n=2}^m\kappa_n\right)D(U)^2+\kappa_1D(U)^4.
\end{aligned}
\end{align*}
Now we set
\[
\kappa_-=2\kappa_1-\sqrt{d}\sum_{n=2}^m\kappa_n,\quad \kappa_+=2\kappa_1+\sqrt{d}\sum_{n=2}^m\kappa_n,
\]
then we can express above estimate as follows
\begin{align*}\label{E-7}
-\kappa_+D(U)^2-\kappa_1D(U)^4\leq\frac{d}{dt}D(U)^2\leq-\kappa_-D(U)^2+\kappa_1D(U)^4.
\end{align*}
and assume that $\kappa_->0$. This implies $\kappa_1$ must be positive and $\kappa_n$ with $n>1$ can be negative. 
\end{proof}

\subsection{Relaxation estimate} \label{sec:5.2}
In this subsection, we derive decay estimates for ${\mathcal D}(U)$. For this, we first present estimates on the Riccati type differential inequalities.
\begin{lemma} \label{L5.2}
Suppose that a differential inequality $X$ satisfies a differential inequality:
\[
-\kappa_+X-\kappa_1X^2\leq \frac{d}{dt}X\leq -\kappa_-X+\kappa_1X^2, \quad t > 0, \qquad 0\leq X(0)<\frac{\kappa_-}{\kappa_1}.
\]
Then we have following inequality.
\[
\frac{\frac{\kappa_+}{\kappa_1}X(0)}{e^{\kappa_+t}(X(0)+\kappa_+/\kappa_1)}\leq X(t)\leq  \frac{\kappa_-X(0)/\kappa_1}{e^{\kappa_-t}(\kappa_-/\kappa_1-X(0))+X(0)}
\]
\end{lemma}
\begin{proof}
By direct estimates, one has
\begin{align*}
-\kappa_1X\left(X+\frac{\kappa_+}{\kappa_1}\right)\leq \frac{d}{dt}X\leq -\kappa_1X\left(\frac{\kappa_-}{\kappa_1}-X\right).
\end{align*}
For the lower bound estimate, we use the L.H.S. of the above differential inequality to get 
\[
-\kappa_+\leq \frac{\dot{X}}{X}-\frac{\dot{X}}{X+\frac{\kappa_+}{\kappa_1}}
\]
This yields
\[
\frac{{\kappa_+}X(0)/{\kappa_1}}{e^{\kappa_+t}(X(0)+\kappa_+/\kappa_1)-X(0)}\leq X(t).
\]
Similarly, one has 
\[
X(t)\leq \frac{\kappa_-X(0)/\kappa_1}{e^{\kappa_-t}(\kappa_-/\kappa_1-X(0))+X(0)}.
\]
\end{proof}
Finally, Lemma \ref{L5.1} and Lemma \ref{L5.2} imply the exponential decay estimate of relative states. 
\begin{theorem} \label{T5.1}
Suppose that coupling strengths and initial data satisfy
\[ \kappa_-=2\kappa_1-\sqrt{d}\sum_{n=2}^m\kappa_n> 0 \quad \mbox{and} \quad \max_{1 \leq i, j \leq N}\|U^{in}_i-U^{in}_j\|_F^2< \frac{\kappa_-}{\kappa_1},  \]
and let $\{U_i\}$ be a global solution of system \eqref{D-1}. Then we have
\[
{\mathcal O}(e^{-\kappa_+t})\leq \|U_i(t)-U_j(t)\|_F^2\leq {\mathcal O}(e^{-\kappa_-t}), \quad i, j = 1, \cdots, N.
\]
\end{theorem}

\subsection{Extension to a heterogeneous ensemble} \label{sec:5.3}
For a heterogeneous ensemble,  we can extend the generalized Lohe matrix model \eqref{D-1} by adding $H_j$ to the R.H.S. of \eqref{D-1}:
\begin{align}
\begin{aligned} \label{E-8}
&{\mathrm i} \dot{U}_j U_j^{\dagger} = H_j +  \sum_{n=1}^m\frac{{\mathrm i} \kappa_n}{2}(\underbrace{U_cU_c^\dagger U_c\cdots U_c^\dagger U_c}_{2n-1} U_j^{\dagger}-U_j \underbrace{U_c^\dagger U_cU_c^\dagger \cdots U_cU_c^\dagger}_{2n-1}), ~~t >0,\\
& U_j(0) =U_j^{in} \in\mathbb{U}(d), \quad j = 1, \cdots, N,
\end{aligned}
\end{align}
where $H_j$ is a Hermitian matrix with $H_j^* = H_j$. In this case, it is easy to see that 
\[ \frac{d}{dt} U_j^{\dagger} U_j = 0, \quad j = 1, \cdots, N. \]
Hence, we have
\[ U_j^{\dagger} U_j = I_d, \quad j  = 1, \cdots, N. \]
In this case, system \eqref{E-8} becomes
\begin{align}\label{E-9}
\begin{cases}
\dot{U}_j =-\mathrm{i}H_j +\displaystyle\sum_{n=1}^m\frac{\kappa_n}{2}(\underbrace{U_cU_c^\dagger U_c\cdots U_c^\dagger U_c}_{2n-1}-U_j \underbrace{U_c^\dagger U_cU_c^\dagger \cdots U_cU_c^\dagger}_{2n-1}U_j),\\
U_j(0) =U_j^0\in\mathbb{U}(d).
\end{cases}
\end{align}
For an ensemble $\{H_j \}$, we set 
\[
D(H):= \max_{i,j} \|H_i-H_j\|_F.
\]
We use the same argument as in \cite{H-R}, one has following estimate:
\begin{align*} 
\begin{aligned} 
&- 2D(H)D(U)-\kappa_+D(U)^2-\kappa_1D(U)^4\leq\frac{d}{dt}D(U)^2 \\
& \hspace{4cm} \leq 2D(H)D(U)-\kappa_-D(U)^2+\kappa_1D(U)^4.
\end{aligned}
\end{align*}
This yields
\begin{equation} \label{E-10}
-D(H)-\frac{\kappa_+}{2}D(U)-\frac{\kappa_1}{2}D(U)^3\leq\frac{d}{dt}D(U)\leq D(H)-\frac{\kappa_-}{2}D(U)+\frac{\kappa_1}{2}D(U)^3.
\end{equation}
\begin{theorem}
Suppose that system parameters and initial data satisfy
\begin{align*}
\begin{aligned}
& \kappa_-=2\kappa_1-\sqrt{d}\sum_{n=2}^m\kappa_n>0,\quad D(H)<\frac{1}{3}\sqrt{\frac{\kappa_-^3}{3\kappa_1}}, \\
& U_j^{in \dagger} U_j^{in} = I_d, \quad j = 1, \cdots, N \quad  \mbox{and} \quad 
D(U^{in})<\rho=\frac{2D(H)}{\kappa_1}.
\end{aligned}
\end{align*}
Then, for a global solution $\{U_i\}$ to \eqref{E-8}, we have 
\[
\lim_{\kappa_1\rightarrow0}\limsup_{t\rightarrow\infty}D(U)=0.
\]
\end{theorem}
\begin{proof}
For the decay estimate of \eqref{E-10}, we set 
\[
f(x) :=D(H)-\frac{\kappa_-}{2}x+\frac{\kappa_1}{2}x^3.
\]
Then we have
\[
\frac{d}{dt}D(U)\leq f(D(U)).
\]
Now we want to analyze the graph of the $f(x)$ defined on $x\geq0$. Let $\zeta$ be the positive solution of the $f'(x)$. Since
\[
f'(x)=-\frac{\kappa_-}{2}+\frac{3\kappa_1}{2}x^2,
\]
there is only one positive solution and only one negative solution. Then the global minimum of $f(x)$ with the range $x\geq0$ is at $x=\zeta$ with
\[
\zeta=\sqrt{\frac{\kappa_-}{3\kappa_1}}.
\]
So the global minimum is
\[
f(x)\leq f(\zeta)=D(H)-\frac{1}{3}\sqrt{\frac{\kappa_-^3}{3\kappa_1}}.
\]
From the assumption
\[
D(H)<\frac{1}{3}\sqrt{\frac{\kappa_-^3}{3\kappa_1}},
\]
we have two distinct positive solutions $\eta_1$ and $\eta_2$ of $f(x)=0$ with $\eta_1<\eta_2$. Then, we know
\[
f(x)>0\quad  \mbox{at}\quad  x<\eta_1,\quad x>\eta_2;\qquad f(x)<0\quad \mbox{at}\quad \eta_1<x<\eta_2.
\]
If the initial data satisfies
\[
D(U^{in})<\eta_2,
\]
then
\[
\limsup_{t\rightarrow\infty}D(U)\leq \eta_1.
\]
Now we want to find the estimate for $\eta_1$. Since 
\[
f''(x)=3\kappa_1x\geq0\quad \forall x\geq0
\]
If we draw the tangent line $l$ at $(0, D(H))$ on the graph of $y=f(x)$, then $l$ intersects with $x$-axis at $(\rho, 0)$ with 
\[
0<\eta_1<\rho.
\]
Since $\rho=\frac{2D(H)}{\kappa_1}$, we have
\[
\limsup_{t\rightarrow\infty}D(U)\leq \frac{2D(H)}{\kappa_1}.
\]
Finally, we have the practical synchronization:
\[
\lim_{\kappa_1\rightarrow0}\limsup_{t\rightarrow\infty}D(U)=0.
\]
\end{proof}

\section{Conclusion} \label{sec:6}
\setcounter{equation}{0}
In this paper, we have derived a generalized Lohe matrix model with a higher-order polynomial coupling via a gradient flow approach. In \cite{H-K-R2}, the first author and his collaborator have shown that the Lohe matrix model can cast as a gradient flow with a quadratic potential on the unitary group. In the original Lohe's works \cite{Lo-1, Lo-2}, the quadratic coupling is not justified a priori. Hence it is not clear why Lohe employed a cubic interaction for the evolution of the state. In authors' earlier work on the Lohe tensor model which is a high-dimensional generalization of the Lohe matrix model, couplings can be allowed to include odd high-order ones To incorporate higher-order couplings, we use a gradient flow approach to derive a generalized Lohe matrix model with higher-order couplings by employing a higher-order potential and gradient flow approach altogether. For the proposed model, we presented a sufficient framework for the emergent dynamics in terms of system parameters and initial data. Our gradient flow approach is restricted to a homogeneous ensemble. Hence, its extension to the heterogeneous remains unsolved at present, and we leave it for a future work. 
\appendix

\newpage


\begin{thebibliography}{99}
\bibitem{A-B} Acebron, J. A., Bonilla, L. L., P\'{e}rez Vicente, C. J. P., Ritort, F. and Spigler, R.: \textit{The Kuramoto model: A simple paradigm for synchronization phenomena.} Rev. Mod. Phys. {\bf 77} (2005), 137-185.

%\bibitem{A-R} Aeyels,  D. and Rogge, J.: \textit{Stability of phase locking and existence of entrainment in networks of globally coupled oscillators.} Prog. Theor. Phys. {\bf112} (2004), 921-941.

\bibitem{Ba} Barb$\check{a}$lat, I.: \textit{Syst$\grave{e}$mes d$\acute{e}$quations diff$\acute{e}$rentielles dÕoscillations non Lin$\acute{e}$aires}. Rev. Math. Pures Appl. {\bf 4} (1959), 267-270.   

\bibitem{B-C-M} Benedetto, D., Caglioti, E. and Montemagno, U.: \textit{On the complete phase synchronization for the Kuramoto model in the mean-field limit}. Commun. Math. Sci. {\bf 13} (2015), 1775-1786.

% \bibitem{Be} Berger, M.: \textit{Geometry I.} Translated from the 1977 French original by M. Cole and S. Levy. Fourth printing of the 1987 English translation  Universitext. Springer-Verlag, Berlin, 2009.

\bibitem{B-B} Buck, J. and  Buck, E.: \textit{Biology of synchronous flashing of fireflies}. Nature {\bf 211} (1966), 562-564.

%\bibitem{C-C-H} Chi, D., Choi, S.-H. and Ha, S.-Y.: \textit{Emergent behaviors of a holonomic particle system on a sphere.} J. Math. Phys. {\bf 55} (2014), 052703.
%
%\bibitem{C-E-M} Chen, B., Engelbrecht, J. R. and Mirollo, R. L.: \textit{Hyperbolic geometry of Kuramoto oscillator networks}. J. Phys. A {\bf 50} (2017), 355101.

\bibitem{C-H5} Choi, S.-H. and Ha, S.-Y.: \textit{Complete entrainment of Lohe oscillators under attractive and repulsive couplings}. SIAM. J. App. Dyn. Syst. {\bf 13} (2013), 1417-1441.
%
%
%\bibitem{C-H4} Choi, S.-H. and Ha, S.-Y.: \textit{Quantum synchronization of the Sch\"{o}dinger-Lohe model.} J. Phys. A: Math. Theor. {\bf 47} (2014), 355104.
%
%
%\bibitem{C-H1} Choi, S.-H. and Ha, S.-Y.: \textit{Emergent behaviors of quantum Lohe oscillators with all-to-all couplings.} J. Nonlinear Sci. {\bf 25} (2015), 1257-1283.
%
%\bibitem{C-H3} Choi, S.-H. and Ha, S.-Y.:  \textit{Large-time dynamics of the asymptotic Lohe model with a small-time delay.} J. Phys. A: Math. Theor. {\bf 48} (2015), 425101.
%
%
%\bibitem{C-H2} Choi, S.-H. and Ha, S.-Y.: \textit{Time-delayed interactions and synchronization of identical Lohe oscillators.} Quart. Appl. Math {\bf 74} (2016), 205203.


\bibitem {C-H-J-K} Choi, Y., Ha, S.-Y., Jung, S. and Kim, Y.:  \textit{Asymptotic formation and orbital stability of phase-locked states for the Kuramoto model.} Phys. D {\bf 241} (2012), 735-754.

\bibitem{C-S} Chopra, N. and Spong, M. W.: \textit{On exponential synchronization of Kuramoto oscillators}. IEEE Trans. Automatic Control {\bf 54} (2009), 353-357.

%\bibitem{Da} Daido, H.: \textit{Quasientrainment and slow relaxation in a population of oscillators with random and frustrated interactions.} Phys. Rev. Lett. {\bf 68} (1992), 1073-1076.

\bibitem{D} DeVille, L.: \textit{Synchronization and stability for quantum Kuramoto.} J. Stat. Phys. {\bf174} (2019), 160-187. 

\bibitem{D-X} Dong, J.-G. and Xue, X.: \textit{Synchronization analysis of Kuramoto oscillators}. Commun. Math. Sci. {\bf 11} (2013), 465-480.

\bibitem {D-B} D\"{o}rfler, F. and Bullo, F.: \textit{On the critical coupling for Kuramoto oscillators.} SIAM. J. Appl. Dyn. Syst. {\bf 10} (2011), 1070-1099. 


%\bibitem{D-B0} D\"{o}rfler, F. and Bullo, F.: \textit{Exploring synchronization in complex oscillator networks.} IEEE 51st Annual Conference on Decision and Control (CDC) (2012), 7157-7170.


\bibitem{D-B1} D\"{o}rfler, F. and Bullo, F.: \textit{Synchronization in complex networks of phase oscillators: A survey. } Automatica {\bf 50} (2014), 1539-1564.


%\bibitem{H-K-L}  Ha, S.-Y., Kim, Y. and Li, Z.: \textit{Asymptotic synchronization behavior of Kuramoto type models with frustrations.} Netw. and Heterog. Media {\bf 9} (2014), 33-64.
%
%
%\bibitem{H-K-L1} Ha, S.-Y., Kim, Y. and Li, Z.:  \textit{Large-Time Dynamics of Kuramoto Oscillators under the Effects of Inertia and Frustration.} SIAM J. Appl. Dyn. Syst. {\bf 13} (2014), 466-492.

\bibitem{H-K-R0} Ha, S.-Y., Ko, D. and Ryoo, S. W.: \textit{On the relaxation dynamics of Lohe oscillators on some Riemannian manifolds.} J. Stat. Phys. {\bf 172} (2018), 1427-1478.

\bibitem{H-K-R1} Ha, S.-Y., Kim, H. W. and Ryoo, S. W.: \textit{Emergence of phase-locked states for the Kuramoto model in a large coupling regime}. Commun. Math. Sci.  {\bf 14} (2016), 1073-1091.

\bibitem{H-K-R2} Ha, S.-Y., Ko, D. and Ryoo, S. W.: \textit{Emergent dynamics of a generalized Lohe model on some class of Lie groups.} J. Stat. Phys. {\bf168} (2017), 171-207. 

\bibitem{H-P} Ha, S.-Y. and Park, H.: \textit{Emergent behaviors of Lohe tensor flocks}. J. Stat. Phys. {\bf 178} (2020), 1268-1292.

\bibitem{H-R} Ha, S.-Y. and Ryoo, S. W.: \textit{On the emergence and orbital stability of phase-locked states for the Lohe model.} J. Stat. Phys. {\bf 163} (2016), 411-439.

\bibitem{H-L-X}  Ha, S.-Y., Li, Z. and Xue, X.: \textit{Formation of phase-locked states in a population of locally interacting Kuramoto oscillators.} J. Differential Equations {\bf 255} (2013), 3053-3070.

\bibitem{H-R} Ha, S.-Y. and Ryoo, S. W.: \textit{On the emergence and orbital stability of phase-locked states for the Lohe model} J. Stat. Phys. \textbf{163} (2016), 411-439.

%\bibitem{J-M-B} Jadbabaie, A., Motee, N. and Barahona M.: \textit{On the stability of the Kuramoto model of coupled nonlinear oscillators}. Proceedings of the American Control Conference (2004), 4296-4301.
%
%\bibitem{Ki} Kimble, H. J.: \textit{The quantum internet.}  Nature {\bf 453} (2008), 1023-1030.

\bibitem{Ku1} Kuramoto, Y.: \textit{Chemical oscillations, waves and turbulence.} Springer-Verlag, Berlin, 1984.

\bibitem{Ku2} Kuramoto, Y.: \textit{International symposium on mathematical problems in mathematical physics}. Lecture Notes Theor. Phys. {\bf 30} (1975), 420.

%\bibitem{Le} Levnaji\'{c}, Z.: \textit{Emergent multistability and frustration in phase-repulsive networks of oscillators.} Phys. Rev. E {\bf 84} (2011), 016231.
%
%\bibitem{L-H} Li, Z. and Ha, S.-Y.: \textit{Uniqueness and well-ordering of emergent phase-locked states for the Kuramoto model with frustration and inertia.} Math. Models Methods Appl. Sci. {\bf 26} (2016), 357-382.

\bibitem{Lo-6} Lohe, M. A.: \textit{Systems of matrix Riccati equations, linear fractional transformations, partial integrability and synchronization.} J. Math. Phys. {\bf 60} (2019), 072701. 

\bibitem{Lo-5} Lohe, M. A.: \textit{Higher-dimensional generalizations of the Watanabe-Strogatz transform for vector models for synchronization}. J. Phys. A: Math. Theor. {\bf51} (2018), 225101. 

\bibitem{Lo-1} Lohe, M. A.: \textit{Non-abelian Kuramoto model and synchronization.} J. Phys. A: Math. Theor. {\bf 42} (2009), 395101.

\bibitem{Lo-2} Lohe, M. A.: \textit{Quantum synchronization over quantum networks.} J. Phys. A: Math. Theor. {\bf 43} (2010), 465301.

%\bibitem{Lo-4} Lohe, M. A.: \textit{The WS transform for the Kuramoto model with distributed amplitudes, phase lag and time delay.} J. Phys. A: Math. Theor. {\bf50} (2017), 505101. 

%\bibitem{M-M-S}  Marvel, S. A., Mirollo, R. and Strogatz, S. H.: \textit{Identical phase oscillators with global sinusoidal coupling evolve by M\"{o}bius group action}. Chaos {\bf19} (2009), 043104.
%
%\bibitem{M-S3} Mirollo, R. and Strogatz, S. H.: \textit{Stability of incoherence in a population of coupled oscillators.} J. Stat. Phys.
%{\bf 63} (1991), 613-635.
%
%\bibitem{M-S2} Mirollo, R. and Strogatz, S. H.: \textit{The spectrum of the locked state for the Kuramoto model of coupled oscillators.}  Phys. D {\bf 205} (2005), 249-266.
%
%
%\bibitem{M-S1} Mirollo, R. and Strogatz, S. H.: \textit{The spectrum of the partially locked state for the Kuramoto model.} J. Nonlinear Sci. {\bf 17} (2007), 309-347.

\bibitem{O} Olfati-Saber, R.: \textit{Swarms on sphere: a programmable swarm with synchronous behaviors like oscillator networks}. IEEE 45th Conference on Decision and Control (CDC) (2006), 5060-5066.


%\bibitem{P-R} Park, K, Rhee, S. W. and Choi, M. Y.: \textit{Glassy synchronization in the network of oscillators with random phase shift.} Phys. Rev. E {\bf 57} (1998), 5030.

\bibitem{Pe} Peskin, C. S.: \textit{Mathematical aspects of heart physiology}. Courant Institute of Mathematical Sciences, New York, 1975.

\bibitem{P-R} Pikovsky, A., Rosenblum, M. and Kurths, J.: \textit{Synchronization: A universal concept in
nonlinear sciences}. Cambridge University Press, Cambridge, 2001.

%\bibitem{S-K} Sakaguchi, H. and Kuramoto, Y.: \textit{A soluble active rotator model showing phase transitions via mutual entrainment}. Progr. Theoret. Phys., {\bf 76} (1986), 576-581.
%
%\bibitem{S-W} Watanabe, S. and Strogatz, S. H.: \textit{Constants of motion for superconducting Josephson arrays}. Phys. D {\bf 74} (1994), 197-253.

\bibitem{St} Strogatz, S. H.: \textit{From Kuramoto to Crawford: exploring the onset of synchronization in populations of coupled oscillators.} Phys. D {\bf 143} (2000), 1-20.

\bibitem{V-M1} Verwoerd, M. and Mason, O.: \textit{On computing the critical coupling coefficient for the Kuramoto model on a complete bipartite graph.} SIAM J. Appl. Dyn. Syst. {\bf 8} (2009), 417-453.

\bibitem{V-M2} Verwoerd, M. and Mason, O.: \textit{Global phase-locking in finite populations of phase-coupled oscillators.} SIAM J. Appl. Dyn. Syst. {\bf 7} (2008), 134-160.

\bibitem{Wi2} Winfree, A. T.: \textit{Biological rhythms and the behavior of populations of coupled oscillators.} J. Theor. Biol. {\bf 16} (1967), 15-42.

\bibitem{Wi1} Winfree, A. T.: \textit{The geometry of biological time}. Springer, New York, 1980.

%\bibitem{Zh} Zheng, Z. G.: \textit{Frustration effect on synchronization and chaos in coupled oscillators.} Chin. Phys. Soc. {\bf 10} (2001), 703-707.

\end{thebibliography}
\end{document}